\documentclass[showpacs,nofootinbib,preprintnumbers,preprint]{revtex4}

\usepackage{amsmath}
\usepackage{amsfonts}
\usepackage{amssymb}
\usepackage{mathrsfs}
\usepackage{graphicx}
\usepackage{hyperref}
\usepackage[caption=false]{subfig}
\usepackage[utf8]{inputenc}

\usepackage{tikz-cd}
\usepackage{amsmath}
\usetikzlibrary{cd}
\usepackage{blindtext, subfig}

\graphicspath{{PR/}}

\newcommand{\RP}{\right\rangle}
\newcommand{\LP}{\left\langle }

\newtheorem{theorem}{Theorem}
\newtheorem{corollary}{Corollary}
\newtheorem{proposition}{Proposition}
\newtheorem{definition}{Definition}
\newtheorem{remark}{Remark}
\newenvironment{proof}{\noindent {\textit{Proof}}
}{$\Box$\medskip}

\begin{document}

\title{Photon regions and umbilic conditions in stationary axisymmetric spacetimes}

\author{K.V.~Kobialko} \email{kobyalkokv@yandex.ru} \author{D. V. Gal'tsov} \email{galtsov@phys.msu.ru}
\affiliation{Faculty of Physics, Moscow State University, 119899, Moscow, Russia}








\begin{abstract}
Photon region (PR) in the strong gravitational field is defined as a compact region where photons can travel endlessly without going to infinity or disappearing at the event horizon. In Schwarzschild metric PR degenerates to the two-dimensional photon sphere $r=3r_g/2$ where closed circular photon orbits are located. 
The photon sphere as a three-dimensional hypersurface in spacetime is umbilic (its second quadratic form is pure trace).
In Kerr metric the equatorial circular orbits have different radii for prograde, $r_p$, and retrograde, $r_r$, motion (where $r$ is Boyer-Lindquist radial variable), while   for $r_p<r<r_r$ the spherical orbits with constant $r$ exist which are no more planar, but filling some spheres. These spheres, however, do not correspond to umbilic hypersurfaces. In more general stationary axisymmetric spacetimes not allowing for complete integration of geodesic equations, the numerical integration show the existence of PR as well, but the underlying geometric structure was not fully identified so far. Here we suggest geometric description of PR in generic stationary axisymmetric spacetimes, showing that PR can be foliated by  {\em partially umbilic hypersurfaces}, such that the umbilic condition holds for classes of orbits defined by the foliation parameter. New formalism opens a way of analytic description of PR in stationary axisymmetric spacetimes with non-separable geodesic equations.
\end{abstract}


\maketitle

\setcounter{page}{2}

\section{Introduction}
\label{intro}
Formation of shadows of spherically symmetric black holes is closely related to photon spheres, where closed circular photon orbits are located. The shadow corresponds to the set of light rays that inspiral asymptotically onto the light rings \cite{Virbhadra:1999nm}. In the Schwarzcshild case the photon sphere has the radius 
$r=3r_g/2$ and it is densely filled by light rings located at different values of the polar angle $\theta$.

In non-spherical static spacetimes, properties of the photon spheres can be  shared by the {\em photon surfaces} of non-spherical form. In this case \cite {Claudel:2000yi} one deals with  a closed timelike hypersurface such that any null geodesic initially tangent to it remains in it forever. Several examples of spacetimes have been found that allow for non-spherical  photon  surfaces, which are not necessarily asymptotically flat
(vacuum C-metric, Melvin's solution of Einstein-Maxwell theory and its generalizations including the dilaton field \cite{Gibbons}).

Mathematically, an important property of the photon surfaces is established by the theorem asserting that these  are  conformally invariant and totally umbilical hypersurfaces in spacetime \cite{Chen,Okumura,Senovilla:2011np}. This means that their second fundamental form is pure trace, i.e. is proportional to the induced metric.   
\begin{equation}
\sigma(u,v)=H \LP u,v\RP, \quad \forall u,v \in TS.
\label{01}
\end{equation} 
This property may serve a constructive definition of PS, instead of direct solving the geodesic equations. It is especially useful in the cases when the geodesic equations are non-separable, so no analytic solution can be found.

Situation becomes more complicated in stationary axisymmetric spacetimes with rotation, when circular orbits typically exist in the equatorial plane in presence of $Z_2$ symmetry
$\theta\to \pi-\theta$. In the Kerr metric the prograde and retrograde equatorial light
rings have different radii $r_p,\,r_r,\; r_p<r_r$, where $r$ is the Boyer-Lindquist coordinate. Due to existence of the Carter integral, the geodesic equations give rise to independent equations for $r$ and $\theta$ motion, from which one finds that the orbits with constant $r$ exist in the interval $r_p<r<r_r$ for which $\theta$ oscillates between some bounds, so that the orbits lie on the (part of) some spherical surface (spherical orbits, \cite{Wilkins:1972rs,Teo}). The whole set of these surfaces constitute a volume region known as Photon region (PR) \cite{Grenzebach,Grenzebach:2015oea,Grover:2017mhm}. The photon region is equally important in determination of black hole shadow as the photon sphere in the spherically symmetric case (for recent review of strong gravitational lensing and shadows see \cite{Cunha:2018acu,Shipley:2019kfq})
 
Obviously, the existence of the photon sphere is related to spherical symmetry of spacetime. It is worth noting, that the photon sphere is not destroyed by the Newman-Unti-Tamburino (NUT) parameter, in which case the $so(3)$ algebra still holds  locally, though metric is already non-static. With this exception, stationary metrics with true rotation do not admits photon spheres or more general photon surfaces. In static spacetime  various uniqueness theorems were formulated in which an assumption of the existence of a regular horizon was replaced by an assumption of  existence of a photon sphere \cite{Cederbaum,Yazadjiev:2015hda,Yazadjiev:2015mta,Yazadjiev:2015jza,Rogatko,Cederbaumo,Yoshino:2016kgi}. No such general results are available for stationary spacetimes. So the problem of optical characterization of stationary metrics which we discuss in this paper remains relevant. Mention an interesting relation between the separability of spacetime and properties of the circular and the spherical photon orbits discovered recently. Namely, a spacetime is non-separable, if there exist an equatorial circular orbit and, at the same time, there are no spherical orbits beyond the equatorial plane \cite{Pappas:2018opz,Glampedakis:2018blj}. This property may serve a diagnostic of the non-Kerr nature of spacetime.
 
As is well known, in rotating spacetimes the photon orbits with constant Boyer-Lindquist radius
may exist as well (e.g. spherical orbits in Kerr  \cite{Wilkins:1972rs,Teo}), but they do not fill densely the photon spheres, since their existence requires certain relation between the constants of motion. Such orbits fill the three-dimensional volumes --- the {\em photon regions} \cite{Grenzebach,Grenzebach:2015oea}. The corresponding region of spacetime can be interpreted as a set of non-closed timelike hypersurfaces, parameterized by the value of the azimuthal impact parameter $\rho=L/E$, where $L,\, E$ are the motion integrals corresponding to timelike and azimuthal Killing vectors \cite{Galtsov:2019bty,Galtsov:2019fzq}.

In more general stationary axisymmetric spacetimes, the photon orbits which fill some compact region  were called {\em fundamental photon orbits} (FPO) \cite{Cunha:2017eoe}. Their explicit determination, however, become difficult if geodesic equations are non-separable. In such cases the phase space can have chaotic domains \cite{Cornish:1996de,Cunha:2016bjh,Semerak:2012dw, Shipley:2016omi,Cunha:2018gql} and analytical determination of PRs was not given so far.
The purpose of the present paper is to fill this gap. We suggest a new definition of {\em partially umbilic} hypersurfaces relaxing the condition (\ref{01}). Namely, one can impose the condition (\ref{01}) not on {\em all} vectors from the tangent space (TS), but only on some subset of TS, specified by the azimuthal impact parameter. In the Kerr metric, the sphere on which the spherical photon orbits wind are just the spatial sections of such hypersurfaces. In more general cases their spatial sections may have arbitrary shape but must be compact. Having definition of PR in terms of hypersurface but not photon orbits helps to find them in the case of geodesically non-separable spacetimes. Note that our method has common features with the recent idea to analyse PR in the Kerr metric from the point of view of the structure of the tangent  put forward in \cite{Cederbaum:2019vwv}       

The paper consists of two parts. The Section \ref{S1} contains a geometric formulation of the concept of a fundamental photon submanifolds. In Subsection \ref{SS1} we describe splitting of the phase space into sectors specified by the azimuthal impact parameter $\rho=L/E$  of the geodesics. Then in Subsection \ref{SS2} we introduce  the concept of partial umbilical submanifolds on a limited phase space, and define the {\em fundamental photon hypersurfaces} (FPH) on which the FPO are located.
Then in Subsection \ref{SS3} we give structure equations for the main curvatures of the spatial section of FPH and show that such section for convex FPH with  $\rho=0$ have the topology of the sphere $ \mathbb S^2 $ under some assumptions  on the tensor energy-momentum of matter.
The Section \ref{S2} contains geometrical definition of photon region (Subsection \ref{SS5}) and applications of the new formalism to three exact solutions: Kerr in Subsection \ref{SS7}, Zipoy-Voorhees with $\delta=2$ \cite{Zipoy,Voorhees:1971wh,Griffiths,Kodama:2003ch} in Subsection \ref{SS8} , and Tomimatsu-Sato  $\delta=2$ solution \cite{Kodama:2003ch} in Subsection \ref{SS9}. It is shown that they have PRs of three different types, and as a consequence, different  optical shadows  \cite{Grenzebach, Abdikamalov:2019ztb, Bambi:2010hf,Galtsov:2019fzq}). For some of them the mapping $\rho\rightarrow PR$ not always univalent, contrary to the Kerr case  \cite{Grenzebach,Grenzebach:2015oea}.

\section{Geometry of the Fundamenthal Photon Submanifolds} 
\label{S1}

\setcounter{equation}{0}

\subsection{Geometry of the phase space} 
\label{SS1}
Consider stationary axisymmetric spacetime containing the PR sector defined as a compact region containing worldlines of photons  moving indefinitely along periodic or non-periodic orbits (FPOs). Any such orbit is characterized by two integrals of motion, $E,\,L$ whose ratio $\rho=L/E$ is an azimuthal impact parameter. In view of axial symmetry, the worldlines with fixed $\rho$, forming a set of FPOs, lie on some hypersurface in spacetime which we will call fundametal photon hypersurface (FPH). The totality of FPOs, filling the entire PR, will be the union of such FPH with different $\rho$. A further step consists in considering the corresponding structures in the phase space of the geodesic system. A similar design has been proposed in  \cite{Cederbaum:2019vwv} to investigate geometry and topology of PR in Kerr gravitational field. Our purpose here is to investigate the phase space structure of PRs in more general stationary axisymmetric spacetimes based on previous work   
 \cite{Grenzebach,Grenzebach:2015oea,Grover:2017mhm,Galtsov:2019bty,Galtsov:2019fzq}. 
 
Let $M$ be an $m$-dimensional Lorentzian manifold \cite{Chen} endowed with a  non-degenerate scalar product  $\LP \; ,\; \RP$, a tangent bundle   $TM$ and supposed to possess two commuting Killing vector fields $\hat{\tau}$ and  $\hat{\varphi}$ defining a stationary axisymmetric spacetime.  Define a one-parametric family of Killing vectors $\left\{\hat{\kappa}_\rho\right\}$ as linear combination:   
\begin{align}
&\hat{\kappa}_\rho=\rho \hat{\tau}+\hat{\varphi},\\
&\LP \hat{\kappa}_\rho,\hat{\kappa}_\rho \RP=\LP\hat{\varphi},\hat{\varphi}\RP+2\LP \hat{\tau},\hat{\varphi} \RP \rho+\LP\hat{\tau},\hat{\tau}\RP\rho^2,
\label{b1} 
\end{align}       
where $\rho$ is still an arbitrary parameter. Generically, the scalar product  $\LP \hat{\kappa}_\rho,\hat{\kappa}_\rho \RP$ has no definite sign on the total manifold $M$ and even can be zero. So it is natural in introduce a partition $M=U^+_\rho\cup U^-_\rho\cup U^0_\rho$, so that
\begin{align}
a) \quad U^+_\rho  \subset M: \LP \hat{\kappa}_\rho,\hat{\kappa}_\rho\RP>0,\\
b) \quad U^-_\rho  \subset M: \LP \hat{\kappa}_\rho,\hat{\kappa}_\rho\RP<0,\\ c) \quad U^0_\rho  \subset M: \LP \hat{\kappa}_\rho,\hat{\kappa}_\rho\RP=0. 
\end{align}
 Now we are able to introduce the key notion of orthogonal complement \cite{Chen} to the set $\left\{\hat{\kappa}_\rho\right\}$, and construct a natural basis on it. This will be useful for subsequent formulation of the theorems and construction of the structure equations for FPHs.   
  
\begin{definition}
We will call an orthogonal complement $\hat{\kappa}_\rho$ a set $\hat{\kappa}^{\bot}_\rho\subset TM$, such that for all $p\in M$ 
\begin{align}
\hat{\kappa}^{\bot}_\rho|_p=\left\{v_p\in T_pM :\LP v_p ,\hat{\kappa}_\rho|_p\RP=0\right\}.
\end{align}
\label{D1}
\end{definition} 
 
\begin{proposition}
If the pull-back of the scalar product $\LP \; ,\; \RP$ on the subspace $\left\{ \hat{\tau},\hat{\varphi}\right\}$ is non-degenerate and has the signature   ($-,+$), in the orthogonal complement $\hat{\kappa}^{\bot}_\rho$ there exists an orthogonal basis $\left\{e_{\rho},e_a\right\}$, such that $\left\{e_a\right\}$ is a ortonormal basis in   $\left\{\hat{\tau},\hat{\varphi}\right\}^{\bot}$, while the vector field $e_{\rho}$

a) everywhere timelike for $U^{+}_\rho$,

b) everywhere spacelike for $U^{-}_\rho$,

c) everywhere null and proportional to $\hat{\kappa}_\rho$ for $U^{0}_\rho$.
\label{P1}
\end{proposition}

\begin{proof}
Introduce the basis in the tangent bundle $TM$ as a set $\left\{\hat{\tau}, \hat{\varphi} ,e_a\right\}$, where $e_a$ an orthonormal basis of the space $\left\{\hat{\tau},\hat{\varphi}\right\}^{\bot}$, which always exists due to non-degeneracy of the pull-back $\LP \; ,\; \RP$ on the subspace $\left\{ \hat{\tau},\hat{\varphi}\right\}$. Clearly, all $e_a$ by construction satisfy the orthogonality condition $\LP \hat{\kappa}_\rho, e_a\RP=0$. Thus it remains to find the vector $e_\rho$ in $\left\{\hat{\tau},\hat{\varphi}\right\}$ such that $\LP \hat{\kappa}_\rho ,e_\rho \RP=0$. The solution is a linear combination of the basis vectors 
\begin{align}
&e_\rho=A_\rho \hat{\tau}-B_\rho \hat{\varphi},  \label{b4a} \\
&A_\rho=\LP\hat{\varphi},\hat{\varphi}\RP +\rho \LP\hat{\tau},\hat{\varphi}\RP, \quad B_\rho=\rho\LP\hat{\tau},\hat{\tau}\RP+\LP\hat{\tau},\hat{\varphi}\RP,\\
&\LP e_\rho  ,e_\rho \RP=\left\{\LP\hat{\tau},\hat{\tau}\RP\LP\hat{\varphi},\hat{\varphi}\RP-\LP\hat{\tau},\hat{\varphi}\RP^2\right\}\LP \hat{\kappa}_\rho,\hat{\kappa}_\rho\RP,
\label{b4} 
\end{align}
where an expression in braces is the determinant of the induced metric on the space $\left\{\hat{\tau},\hat{\varphi}\right\}$ and consequently is always negative in view of the non-degeneracy and the signature $(-,+)$. Thus the statements  a),b) directly follow from the definition of the partition of $M$. To prove c), substitute into (\ref{b4a}) an expression for $\LP \hat{\tau},\hat{\varphi} \RP$ form the condition $\LP \hat{\kappa}_\rho,\hat{\kappa}_\rho\RP=0$. After simple rearrangements we obtain:
\begin{align}
e_\rho=\left(\frac{\LP\hat{\varphi},\hat{\varphi}\RP-\LP\hat{\tau},\hat{\tau}\RP \rho^2}{2\rho}\right)\hat{\kappa}_\rho.
\label{b5} 
\end{align}
Moreover, the numerator is not equal to zero for a nondegenerate restriction of the metric on $\left\{ \hat{\tau},\hat{\varphi}\right\}$. If $\rho=0$, then $\LP\hat{\varphi},\hat{\varphi}\RP=0$ and consequently $e_\rho=\hat{\varphi}=\hat{\kappa}_\rho$.
\end{proof}

\begin{corollary}
If the restriction of the scalar product $\LP \; ,\; \RP$ on the subspace $\left\{ \hat{\tau},\hat{\varphi}\right\}$ is non-degenerate, the orthogonal complement $\hat{\kappa}^{\bot}_\rho$ will be a subbundle in the tangent bundle $TM$ of dimension ${\rm dim}(\hat{\kappa}^{\bot}_\rho)=2m-1$.    
\label{C1}
\end{corollary}

\begin{remark}
In the regions $U^{\pm}_\rho$ the vector field $e_\rho$ can always be normalized, and the set $\left\{e_{\rho},e_a\right\}$ become an orthonormal basis. However, it is impossible to introduce an orthonormal basis on the full  $M$. Indeed, if this were possible, then on the restriction $\hat{\kappa}^{\bot}_\rho$ on $U^{0}_\rho$ there should exist some orthonormal basis. But in this case the restriction $\hat{\kappa}^{\bot}_\rho$ on $U^{0}_\rho$ in non-degenerate, and consequently its orthogonal complement $(\hat{\kappa}^{\bot}_\rho)^{\bot}=\hat{\kappa}_\rho$ is non-denegerate either, which leads to contradiction.  
\label{R1}
\end{remark}

\begin{remark}
If the Killing vectors $\hat{\tau}$ or $\hat{\varphi}$ have a singular point $p\in M$, then the orthogonal complement $\hat{\kappa}^{\bot}_\rho$ is no more a manifold, since the restriction of the metric on $\left\{ \hat{\tau},\hat{\varphi}\right\}$ degenerates, and the dimension of $\hat{\kappa}^{\bot}_\rho$ can change from point to point. For example, if in the singular point $p$ the field $\hat{\varphi}|_p=0$ and $\rho=0$, then $\hat{\kappa}^{\bot}_0|_p=T_pM$ and ${\rm dim}(\hat{\kappa}^{\bot}_\rho)|_p=2m$. Also, in this case $p\in U^{0}_0$.   
\label{R2}
\end{remark}

Let us now explain the physical meaning of the constructions introduced above. Let $\gamma$ be some geodesic on $M$, and  $\dot{\gamma}$ denotes the tangent vector field to $\gamma$. Consider the scalar product 
\begin{align}
\LP \hat{\kappa}_\rho, \dot{\gamma}\RP=\rho \LP\hat{\tau},\dot{\gamma}\RP+\LP\hat{\varphi},\dot{\gamma}\RP.
\label{b6} 
\end{align}
Then the quantity 
\begin{align}
\LP \hat{\kappa}_\rho, \dot{\gamma}\RP=-\rho E+L,
\label{b7}
\end{align}
where $E$ and $L$ the energy and the azimuthal momentum,  will be conserved along the geodesic. Then, if $\rho=L/E$ is an impact parameter for the chosen geodesic, in each point $p\in\gamma$ the orthogonality condition holds $\LP \kappa_\rho|_p, \dot{\gamma}|_p\RP=0$, and  consequently $\dot{\gamma}\in \hat{\kappa}^{\bot}_\rho$.

Conversely, if at a given point $p\in M$ the tangent vector $v_p\in \hat{\kappa}^{\bot}_\rho|_p$, i.e., $\LP \kappa_\rho|_p, v_p\RP=0$, then $v_p$ is a tangent vector to some geodesic $\gamma$ with an impact parameter $\rho$ or $E=L=0$(the trivial case), which always exists and unique at least in some vicinity of $p\in M$ as solution of  ODE with initial conditions $\gamma(0)=p$ and $\dot{\gamma}(0)=v_p$. 

In a Lorentzian manifold, in view of the Proposition \ref{P1}, the null geodesics $\gamma$ can exist only in the domains $U^{0,+}_\rho$, since otherwise the restriction of the scalar product on $\hat{\kappa}^{\bot}_\rho$ would have Euclidean signature. Also notice that the trivial case ($E=0$) is impossible on $U^{0,+}_\rho$ for null and timelike vectors in $\hat{\kappa}^{\bot}_\rho$, is the restriction $\LP \; ,\; \RP$ onto the subspace $\left\{ \hat{\tau},\hat{\varphi}\right\}$ in non-degenerate. Indeed, for all null (timelike) $e_\rho\pm e_a\in \hat{\kappa}^{\bot}_\rho$($e_\rho\in \hat{\kappa}^{\bot}_\rho$) we obtain $\LP\hat{\tau},e_\rho\pm e_a\RP=\LP\hat{\tau},e_\rho\RP$ and 
\begin{align}
-E=\LP\hat{\tau},e_\rho\RP=\left\{\LP\hat{\tau},\hat{\tau}\RP\LP\hat{\varphi},\hat{\varphi}\RP-\LP\hat{\tau},\hat{\varphi}\RP^2\right\}<0.
\end{align}
This completes the proof.

\begin{proposition}
For any geodesic $\gamma_\rho$ with an impact parameter $\rho$ the tangent vector field $\dot{\gamma}_\rho\in \hat{\kappa}^{\bot}_\rho$. If the restriction of the scalar product $\LP \; ,\; \RP$ on the subspace $\left\{ \hat{\tau},\hat{\varphi}\right\}$ is non-degenerate on $U^{0,+}_\rho$, for each null/timelike  $v_p\in\hat{\kappa}^{\bot}_\rho|_p$ the unique null/timelike geodesic $\gamma$ exists with an impact parameter $\rho$ such that $\gamma(0)=v_p$ and $\dot{\gamma}(0)=v_p$.
\label{P2}
\end{proposition}  

Basically, we are interested by closed connected regions admitting geodesics with fixed $\rho$.

\begin{definition}
A causal $\rho$-region $P_\rho$ will be called a closed connected submanifold in $M$ such that $\partial P_\rho\subset U^{0}_\rho$, and $P_\rho/\partial P_\rho\subset U^{+}_\rho$. The region $O_\rho=P_\rho/\partial P_\rho$ will be called $\rho$-accessible. 
By hat $\hat{{}}$ we denote the restriction $\hat{\kappa}^{\bot}_\rho$ on $P_\rho$ and $O_\rho$.
\end{definition}
 
If there are no singular points the causal $\rho$-bundle $\hat{P}_\rho$, its boundary $\partial \hat{P}_\rho$ and the inner region $\hat{O}_\rho$ are subbundles in the restriction $TM$ on the corresponding submanifolds by virtue of Corollary \ref{C1}. From the point of view of geodesics, and, in particular, the FPOs, the region $P_\rho$ represents an accessible region for the  null geodesics in some effective potential  \cite{Cunha:2017eoe,Lukes}. Physical meaning of the causal region $P_\rho$ is that any point can be theoretically observable for any observer in the same region (for geodesics with fixed $\rho$). This causal region may contain spatial infinity (if any) and then will be observable for an asymptotic observer. In some cases, several causal areas may exist, while null geodesics with a given $\rho$ cannot connect one to another. The boundary $\partial P_ \rho $ of the causal region is defined as the branch of the solution of the equation $ \LP \hat{\kappa}_\rho, \hat{\kappa}_\rho\RP = 0 $ and is the set of turning points of null geodesics.

The accessible region $O_\rho$ is a region in $M$ in which there is a stationary observer with a fixed value of the impact parameter $\rho$. The speed $u_\rho$ of such an observer is equal to the normalized value of the vector $e_\rho$ and is written in the canonical form
\begin{align}
u_\rho=N(\hat{\tau}+\Omega_\rho\hat{\varphi}), \quad \LP u_\rho  ,u_\rho \RP=-1, \\ \Omega_\rho=-B_\rho/A_\rho, 
\end{align} 
where $\Omega_\rho$ - and an angular velocity of an observer which depends non-trivially on the point in space for fixed $\rho$, and $N$ is a normalizing function. In particular for $\rho=0$ we obtain ZAMO observer with $\Omega_0=-\LP\hat{\tau},\hat{\varphi}\RP/\LP\hat{\varphi},\hat{\varphi}\RP$.  

\subsection{Fundamenthal photon submanifold}
\label{SS2}
Let $(M,\hat{g})$ and $(S,g)$ be Lorentzian manifolds, of dimension $m$ and $n$ respectively, and $f:S\rightarrow M$ an isometric embedding  \cite{Chen} defining $(S,g)$ as a submanifold (a hypersurface if $n=m-1$ in $(M,\hat{g})$). Let $TS$ be a tangent bundle over $S$, and $V$ - its subbundle. Let $\hat{\nabla}$ and $\nabla$ - be the Levi-Civita connections on $M$ and $S$ respectively. We adopt here the following convention for the second quadratic form  $\sigma$ of the isometric embedding $f$ \cite{Chen,Okumura,Senovilla:2011np}:      
\begin{align}
\hat{\nabla}_u v=\nabla_u v +\sigma(u,v), \quad \forall u,v\in TS,
\label{b56} 
\end{align} 
where $\nabla_u v\in TS$ and $\sigma(u,v)\in TS^{\bot}$, where $TS^{\bot}$ - is a standard orthogonal complement (see, e.g.,  \cite{Chen}).

\begin{definition}
We will call an isometric embedding $f:S\rightarrow M$ invariant, if the Killing vector fields  $\hat{\tau}$ and $\hat{\varphi}$ in $M$ are tangent vector fields to $S$. 
\label{D3}
\end{definition}

For invariant embeddings the Killing vectors of $M$ will be also the Killing vectors on the submanifold  $S$, what can be easily verified projecting the Lei derivative onto $S$. In this case there is a natural correspondence between the pullback of  $\hat{O}_\rho|_{S}$ on $S$ and an intrinsic $\hat{O}^S_\rho$ in the $S$ itself (as well as for  $\hat{P}_\rho|_{S}$)
\begin{align}
\hat{O}_\rho|_{S}=\hat{O}_\rho^S \oplus TS^{\bot},
\end{align} 
since the vector $\hat{\kappa}_\rho$ is tangent to $S$, and the orthogonal vector fields are projected into orthogonal. 

By virtue of the Poincare-Hopf theorem, not any manifold $S$ admits the existence of a smooth tangent vector field $\hat{\varphi}$ without singular points (in particular, vector fields on a sphere $\mathbb S^2$ have at least one singular point, since $\mathbb S^2$ has the Euler number $\chi=2$). We will assume the singular points  $p\in S:\hat{\varphi}|_p=0$ (See the Remark \ref{R2}). In the case of the submanifolds $S$ corresponding to  $\rho=0$, the orthogonal complement  $\hat{\kappa}^{\bot}_0|_p=T_pS$, while the singular points will lie on the boundary    $p\in \partial P^S_0$. Indeed, for an arbitrary vector $v_p\in T_{p}S$ \begin{align}
\LP v_p ,\hat{\kappa}_\rho|_p\RP=\rho \LP v_p , \hat{\tau}|_p\RP=-\rho E=0,
\end{align}
and
\begin{align}
&\LP \hat{\kappa}_0,\hat{\kappa}_0 \RP|_p=\LP\hat{\varphi},\hat{\varphi}\RP|_p=0.
\end{align}  
If $\rho\neq0$, the null tangent vectors with a given $\rho$ must correspond to zero value of the energy $E$. Such singular points will not be considered as the geodesics we are interested in don’t pass through them anyway. In all non-singular points we will always     require the non-degeneracy of the restriction of $\LP \; ,\; \RP$ on the subspace $\left\{ \hat{\tau},\hat{\varphi}\right\}$.
Therefore, in particular, $\hat{O}^S_\rho$ is a subbundle in $TS$, under our assumption that $O^S_\rho$ does not contain singular points, and $S$ for $\rho=0$ can have singular points only on the boundary $\partial P^S_0$.
 
We now define a weakened version of the standard umbilical condition (\ref{01}) \cite{Chen,Okumura,Senovilla:2011np} requiring it to be satisfied only for some subbundle $V$ in the tangent bundle $TS$.

\begin{definition}
A point $p\in S$ will be called a $V$-umbilic point of an isometric embedding $f:S\rightarrow M$ if
\begin{align}
\sigma(u,v)=H \LP u,v\RP, \quad \forall u,v\in V_p,\quad H\in T_pS^{\bot}.
\end{align}  
\end{definition}
A totally $V$-umbilic embedding $f: S\rightarrow M$ is an isometric embedding $V$-umbilic at all points $S$. In particular, every totally umbilical embedding is trivially totally $V$-umbilic for any $V$. We also note that in the general case $H$ appearing in this formula is not the mean curvature of \cite{Chen}. For invariant completely $V$-umbilic embeddings, an important theorem on the behavior of null geodesics holds, generalizing the classical result \cite{Chen,Claudel:2000yi}.  

\begin{theorem}
Every null geodesic $\gamma_\rho$  on an invariant submanifold $S_\rho\subset O_\rho$ (${\rm dim}(S_\rho)>2$) is null geodesic in $M$ if and only if $f_\rho:S_\rho \rightarrow O_\rho$ is an invariant totally $\hat{O}^S_\rho$-umbilic embedding (compare with the analogous statement for totally umbilical surfaces \cite{Chen}).
\label{T1}
\end{theorem}

\begin{proof}
Suppose that $f_\rho:S_\rho \rightarrow O_\rho$  is a invariant totally $\hat{O}^S_\rho$-umbilic embedding. Let $\gamma_\rho$ be a null geodesic with the impact parameter $\rho$ on the invariant submanifold $S_\rho\subset O_\rho$, i.e, $\nabla_{\dot{\gamma}_\rho} \dot{\gamma}_\rho=0$. Consider an arbitrary point $p\in S_\rho$. For a null tangent vector, $\dot{\gamma}_\rho|_p$ Proposition \ref{P2} means that $\dot{\gamma}_\rho|_p\in\hat{O}^S_\rho|_p$. By our assumption, the isometric embedding $f$ is totally $\hat{O}^S_\rho$-umbilic. Then for the null vector $\dot{\gamma}_\rho|_p$ we get $\sigma(\dot{\gamma}_\rho,\dot{\gamma}_\rho)|_p=0$ and therefore according to formula (\ref{b56}) $\hat{\nabla}_{\dot{\gamma}_\rho} \dot{\gamma}_\rho=0$  i.e. $\gamma_\rho$ is a null geodesic with the impact parameter $\rho$ in $M$.

Conversely, let every null geodesic $\gamma_\rho$ on an invariant submanifold  $S_\rho\subset O_\rho$ be a null geodesic in $M$. By Proposition \ref{P2}, for any null $v_\rho|_p\in\hat{O}^S_\rho|_p$, $v_\rho|_p$ 
is the tangent vector to some null geodesic at the point $p$. Thus, for any null vector $v_\rho|_p\in\hat{O}^S_\rho|_p$ we have  $\sigma(v_\rho,v_\rho)|_p=0$.
By virtue of Proposition \ref{P2}, Remark \ref{R1} and the invariance condition Definition \ref{D3}, we can construct an orthonormal basis 
$\left\{e_\rho,e_a\right\}$ in the space $\hat{O}^S_\rho|_p$.
We now consider the set of null vectors $e_\rho\pm e_a$ in $\hat{O}^S_\rho)|_p$.
 By the previously proved  
      $\sigma(e_\rho\pm e_a,e_\rho\pm e_a)|_p=0$, from which we get 
\begin{align}
\sigma(e_\rho,e_a)=0, \quad  \sigma(e_\rho,e_\rho)+\sigma(e_a,e_a)=0.
\label{a31}
\end{align}  
Consider now a null vector  $e_\rho+(e_a+ e_b)/\sqrt{2}$, for which we obtain
\begin{align}
\sigma(e_a,e_b)=0. 
\label{a32}
\end{align} 
\end{proof}
 
\begin{remark}
The first part of the statement of the theorem can be trivially extended to the entire causal region $P_\rho$. However, in the opposite direction this is no longer true, so in $\partial P_\rho$  there is only one isotropic vector $\hat{\kappa}_\rho$.
\label{R3}
\end{remark}

Physical meaning of the theorem is 
that the null geodesics with a given $\rho$ initially touching the spatial section  of the invariant totally $\hat{O}^S_\rho$-umbilic submanifold remain on it for an arbitrarily long time, unless of course they leave it across the boundary. This is a well-known property of a {\em photon sphere} and its generalization - a {\em photon surface} (PS) \cite{Claudel:2000yi}. Thus, we obtain a generalization of the classical definition of the photon  surfaces to the case of a class of geodesics with a fixed impact parameter.

It is useful to obtain an equation for the second fundamental form of the totally $\hat{O}^S_\rho$-umbilic embedding in the original basis $\left\{\hat{\tau},\hat{\varphi}, e_a\right\}$.  First of all, we will agree on the notation. By definition, put $\tilde{\sigma}_{\tau\tau}\equiv\sigma(\tau,\tau)$, etc. if the second fundamental form is calculated on an unnormalized basis and $\sigma_{ab}\equiv\sigma(e_a,e_b)$ on a normalized one. Substituting the explicit expression for $e_\rho$ into (\ref{a31}) and (\ref{a32}),  we get:
\begin{align}
&A_\rho \tilde{\sigma}_{\tau a}-B_\rho \tilde{\sigma}_{\varphi a}=0, \\
 &A^2_\rho \tilde{\sigma}_{\tau\tau}-2 A_\rho B_\rho \tilde{\sigma}_{\tau\varphi}+ B^2_\rho \tilde{\sigma}_{\varphi\varphi} +N^2_\rho\sigma_{aa}=0,\\
& \sigma_{aa}=\sigma_{bb}, \quad \sigma_{ab}=0,
\end{align} 
where $N_\rho=||e_\rho||$ is a norm.
This structural $\hat{O}^S_\rho$-umbilic equation is defined and works both in the ergoregion and in the area of causality violation. If the Killing vectors have a nonzero norm, it is also convenient to introduce a completely normalized basis
  $\left\{e_\tau,e_\varphi, e_a\right\}$
\begin{align}
e_\tau=\hat{\tau}/\tau, \quad e_\varphi=\hat{\varphi}/\varphi, \quad e_a,
\end{align}   
where $\tau=||\hat{\tau}||$ and  $\varphi=||\hat{\varphi}||$.  This can be done in a fairly general situation when there are no ergoregions or areas of non-causality. In this case we will write $\sigma_{\tau\tau}\equiv\sigma(e_\tau,e_\tau)$ etc. By bilinearity, it is obvious that $\tau^2\sigma_{\tau\tau}=\tilde{\sigma}_{\tau\tau}$ etc. 

The notion of a $\hat{O}^S_\rho$-umbilic embedding is however too general (as is the notion of an umbilical surface by itself \cite{Cao:2019vlu}). Generally speaking, these submanifolds are geodesically not complete (in the sense that null geodesics can leave them across the boundary) or have a non-compact spatial section (geodesics can go into the asymptotic region). Moreover, for each $\rho$ there can be an infinite number of them, just as there are an infinite number of umbilical surfaces, but only one photon sphere in the static Schwarzschild \cite{Cederbaum:2019rbv} solution. Therefore,  it is necessary to introduce a more specific definition of fundamental photon submanifolds.
 
\begin{definition}
A fundamental photon submanifold is an invariant isometric embedding of Lorentzian manifolds $f_\rho:S_\rho \rightarrow P_\rho$   with compact spatial section $I$ (see below for a possible way to define the spatial section for the case of a hypersurface) such that:

a) All non-singular internal points  $q\in (S_\rho/\partial S_\rho)\cap O_\rho$ are $\hat{O}^S_\rho$-umbilic.

b) All boundary points $p\in\partial S_\rho$ (if any) lie in $\partial P_\rho $.

c) For all non-singular points $g\in S_\rho \cap \partial  P_\rho$ (both boundary and internal), the condition holds $\sigma(\hat{\kappa}_\rho,\hat{\kappa}_\rho)|_g=0$.

d) All the singular points $o\in S_0\cap\partial P_0$ are umbilical.
\label{D5}
\end{definition}
In the case ${\rm dim}(S_\rho)=m-1$, the fundamental photon submanifold is a timelike fundamental photon  hypersurface  (FPH). In the case ${\rm dim}(S_\rho)=2$, it is the fundamental photon orbit (axially symmetric and lying in $\partial  P_\rho$ - for example, circular equatorial).
 
\begin{proposition}
If $O_\rho$ is connected, then every null geodesic $\gamma_\rho$ at least once touching an arbitrary FP-submanifold $S_\rho$ lies in it completely:
  $\gamma_\rho\subset S_\rho$.    
\end{proposition}

\begin{proof}
Condition a), by virtue of Theorem \ref{T1}, prevents null geodesics from leaving the FP-submanifold at all interior points $q\in (S_\rho/\partial S_\rho)\cap O_\rho$. Condition b) for boundary non-singular points $p\in \partial S_\rho$  prevents the possibility of null geodesics to leave fundamental photon submanifolds through the boundary (if any). Indeed, $\partial P_\rho$ is the set of turning points for null geodesics that can only touch $\partial P_\rho$, and then either go inside the region $O_\rho$ (if $O_\rho$ is connected, then a null geodesic will not go into another connected component) or just stay in $\partial P_\rho$. Condition c) $\sigma(\hat{\kappa}_\rho,\hat{\kappa}_\rho)=0$ ensures the return of null geodesics to a totally $\hat{O}^S_\rho$-umbilic submanifold after reflecting at the turning point (it is enough since there is only one null vector $\hat{\kappa}_\rho$ in $\partial  P_\rho$). In case d) if $o\in S_0$  is a singular point, then $\hat{\kappa}^{\bot}_0|_o=T_oS$. Then, to generalize the proof of Theorem \ref{T1} to this case, we can consider instead of $\left\{e_\rho,e_a\right\}$ an arbitrary non-degenerate orthnormal basis in the complete tangent space $T_{o}S$ which always exists for an isometric embedding \cite{Chen}. Moreover, the point $o$ itself will prove to be umbilical. Since all the singular points of surfaces with $\rho\neq0 $ are not attainable by, the statement is completely proved.
\end{proof}

\begin{remark}
If $O_\rho$ is disconnected, then a null geodesic can, in principle, leave the FP-submanifolds through the boundary lying in $\partial P_\rho $ by passing from one connected component of $O_\rho $ to another.
\end{remark}

From this statement, it is clear that the so-defined fundamental photon submanifolds in the most general case can contain two types of null geodesics:

a) Non-periodic photon  orbits (trapped in the FP-submanifold).

b) Periodic fundamental photon  orbits \cite{Cunha:2017eoe}.

Thus, FP-submanifolds generalize the concept of the latter and give them a new geometric interpretation, providing us with new tools of the theory of submanifolds, which has demonstrated its strength in constructing uniqueness theorems \cite{Cederbaum,Yazadjiev:2015hda,Yazadjiev:2015mta,Yazadjiev:2015jza,Rogatko,Cederbaumo} and analysis of topological properties.
  
\subsection{Fundamenthal Photon Hypersurfaces}
\label{SS3}
We now turn to the study of fundamental photon hypersurfaces ${\rm dim}(S)=m-1$, their spatial section $I$, and the dynamics of null geodesics on them. We first consider a 3-dimensional fundamental photon hypersurface (or even a 3-dimensional submanifold). In the 3-dimensional case, there are a number of strict restrictions on the behavior of null geodesics on the FPH, since at each point there are only two linearly independent null tangent vectors with fixed $\rho$ (and, accordingly, at most two null geodesics $\dot{\gamma}_\rho$ can pass through each point).

Let $\gamma(s)$ be some null geodesic on a 3-dimensional FPH passing through the point $p\in O^S_\rho$ when $s=0$. We introduce locally in a neighborhood of the point $p$ an adapted coordinate system $(\tau,\theta,\varphi)$ such that
\begin{align}
\hat{\tau}=\partial_\tau, \quad \hat{\varphi}=\partial_\varphi, \quad e_a=\partial_\theta.
\end{align} 
We define the projection of the geodesic onto the subspace $(\theta,\varphi)$ as a two-dimensional curve
  $(\gamma_\theta(s),\gamma_\varphi(s))\subset(\theta,\varphi)$, where $-\epsilon<s<\epsilon$. 

\begin{proposition}
At each point $p\in O^S_\rho$ of a 3-dimensional FPH (or even a 3-dimensional FP-submanifold of greater codimension), null geodesics with a given $\rho$ can have at most one intersection/touch or at most one self-intersection/self-touch of projections on the subspace $(\theta,\varphi)$.
\end{proposition} 

\begin{proof}
In the case of a 3-dimensional hypersurface, by Proposition \ref{P1}, at any point $p\in O^S_\rho$ there are only two linearly independent null tangent vectors with given $\rho$, namely
\begin{align}
e_\pm=e_\rho\pm e_a.
\end{align} 
In the adapted basis, the null tangent vectors $\pm e_\pm$ have four projections onto
 $(\theta,\varphi)$, 
\begin{align}
\pm(B_\rho/N_\rho)\partial_\varphi\pm\partial_\theta,
\end{align} 
And accordingly, by virtue of Proposition \ref{P2}, in the case of $B_\rho\neq0$ there may exist an intersection of null geodesics such as a cross, and in the case of $B_\rho=0$, a touch (equal spatial vectors will have different time directions).
\end{proof}   

Note that if the hypersurface has self-intersections, then the number of intersections of null geodesics can also increase. Moreover, there can be an infinite number of intersections at a singular point. For example, in Zipoy-Voorhees metric on the FPH $\rho=0$ there are closed photon orbits in planes perpendicular to the equatorial one and intersecting along the axis of symmetry. At the same time, the point of intersection of the axis of symmetry and the FPH is a special point for the FPH, and the whole family intersects in it. In Kerr, an infinite number of spherical photon orbits intersect at the pole, which, however, they do not lie in any plane and can have single self-intersections.

Let us now explicitly define the notion of the spatial section $I$ for an arbitrary FPH of dimension ${\rm dim}(S)=m-1$. In the case of a stationary axially symmetric space, it is possible to choose the foliation of the manifold $M$ with hypersurfaces $N$ of constant time $\tau$ on which $U(1)$ symmetry is manifest. The Killing vector $\hat{\tau}$ at each point of the hypersurface $N$ then admits decomposition
\cite{Yoshino1}
\begin{equation}
\hat{\tau}=\alpha \hat{n}+\hat{\beta}, \quad \hat{\beta}=-\omega \hat{\varphi}, 
\end{equation}
where $\hat{n}$ is the unit time-like normal to $N$, and $\alpha$ and $\omega$ are the lapse and rotation  functions (ZAMO).
 
We now consider a timelike hypersurface $S$ in $M$ with the normal $\hat{r}$ intersecting $N$ orthogonally in the submanifold $I=S\cap N$ (this means that the normal $\hat{r}$ to $S$ in $M$ coincides with the normal to $I$ in $N$) 
\begin{displaymath}
\begin{tikzcd} 
S  \arrow{rr}{{}^S\sigma}[swap]{\hat{r}}  &  & M  \\ 
I \arrow{u}{\hat{n}}[swap]{} \arrow{rr}{{}^I\sigma}[swap]{\hat{r}} &  & N \arrow{u}{\hat{n}}[swap]{}
\end{tikzcd} 
\end{displaymath} 
For such an intersection, the second fundamental form ${}^S\sigma$ (in the case of hypersurfaces it is simply a scalar function since the normal is unique) the hypersurface $S$ in $M$ is expressed in terms of the second fundamental form ${}^I\sigma$ of $(m-2)$-dimensional spatial section of $I$ in $N$ and lapse function:
\begin{align}
&{}^S\sigma(u,v)={}^I\sigma(u,v),  \\ &{}^S\sigma(u,\hat{n})=\left(\frac{1}{2\alpha}\right)\LP u,\hat{r}(\omega)\hat{\varphi}\RP, \\
&{}^S\sigma(\hat{n},\hat{n})=\frac{\hat{r}(\alpha)}{\alpha}+\frac{1}{\alpha}\LP \hat{n}, \hat{r}(\omega)\hat{\varphi}\RP,
\end{align} 
where $u,v\in TI$. 

We again construct the basis $\left\{e_{\rho},e_a\right\}$, and expand the vector $e_{\rho}$ at the intersection $I = S\cap N$ as follows:
\begin{align}
&N_{\rho}e_{\rho}=A'_\rho \hat{n}-B'_\rho \hat{\varphi}, \\
&A'_\rho=\alpha A_\rho, \quad  B'_\rho=B_\rho+\omega  A_\rho. 
\end{align} 
Then $\hat{O}^S_\rho$-umbilic equation (\ref{a31}), (\ref{a32}) reduce to
\begin{align} 
&A'^2_\rho\left\{\frac{\hat{r}(\alpha)}{\alpha}\right\}-2\rho A'_\rho\left(\frac{\hat{r}(\omega)}{2\alpha}\right)\left\{\LP\hat{\tau},\hat{\tau}\RP\LP\hat{\varphi},\hat{\varphi}\RP-\LP\hat{\tau},\hat{\varphi}\RP^2\right\}+B'^2_\rho\tilde{\sigma}_{\varphi\varphi}+N^2_{\rho}\sigma_{aa}=0,\label{a10} \\
&B'_\rho\tilde{\sigma}_{a\varphi}=0, \quad \sigma_{ab}=0, \quad \sigma_{aa}=\sigma_{bb}.
\label{a33}
\end{align}
Further, we assume that the always mixed components are $\tilde{\sigma}_{a\varphi}=0$ with an appropriate choice of basis. The equation (\ref{a10}) can be simplified even more by requiring the fulfillment of the orthogonality condition $\LP \hat{n},\hat{\varphi}\RP=0$ and the absence of a violation of causality $\LP \hat{\varphi},\hat{\varphi}\RP>0$. In this case, the following relations arise
 \begin{align}
\left\langle \hat{\tau},\hat{\varphi} \right\rangle=-\omega\varphi^2, \quad \LP\hat{\tau},\hat{\tau}\RP=-\alpha^2+\omega^2\varphi^2,
\end{align} 
in particular,
\begin{align}
A'_\rho=\alpha\varphi^2(1-\omega \rho), \quad B'_\rho=-\rho \alpha^2. 
\end{align}
Then the $\hat{O}^S_\rho$-umbilic equation (\ref{a10}) and the causal region inequality $P_\rho$ (\ref{b1}) take the form (we omit the symbol ${}^I$):
\begin{align} 
(1-\omega \rho)^2\left\{\sigma_{aa}-\sigma_{nn}\right\}+2\alpha\rho(1-\omega \rho)\left(\frac{\hat{r}(\omega)}{2\alpha}\right) +\frac{\rho^2 \alpha^2}{\varphi^2}\left\{\sigma_{\varphi\varphi}-\sigma_{aa}\right\}=0, \label{a5}
\end{align}  
and
\begin{align} 
(1-\omega \rho)^2\varphi^2\geq\alpha^2\rho^2.
\end{align}  
$\hat{O}^S_\rho$-umbilic equation (\ref{a33}) can also be rewritten in terms of the principal curvatures  of the spatial section $I$ as
\begin{align} 
&\xi^2_\rho (\lambda_{\varphi}-\lambda_{a})+2 \xi_\rho v+\lambda_{a}-\lambda_{n}=0, \\ &v\equiv\left(\frac{\hat{r}(\omega)}{2\alpha}\right)\varphi, \quad
 \lambda_{n}\equiv-\left(\frac{\hat{r}(\alpha)}{\alpha}\right), \quad \xi_\rho\equiv\frac{\alpha\rho}{(1-\omega \rho)\varphi}, 
\label{a38}
\end{align} 
where $-1\leq\xi_\rho\leq1$ inside the causal region $P_\rho$. These equations are key in the explicit construction of fundamental photon hypersurfaces and are in many respects similar to the equations of transversaly trapping surfaces \cite{Yoshino1,Yoshino:2019dty,Yoshino:2019mqw}. They open the way to the application of Gauss-Codazzi-Ricci \cite{Chen} structural equations for the analysis of topological properties of fundamental photon hypersurfaces and construction of Penrose-type inequalities \cite{Shiromizu:2017ego,Feng:2019zzn,Yang:2019zcn}.
 
An important feature of this equation in the static case ($\omega=0$) is the parity in the parameter $\rho$ (Compare with \cite{Galtsov:2019fzq}). In particular, every subvariety of $S_\rho$ will coincide with $S_{- \rho}$. Note also the possibility of the presence of boundaries at the cross sections of fundamental photon hypersurfaces where the relation $\varphi$-TTS holds \cite{Galtsov:2019fzq}:
\begin{align} 
 \lambda_{\varphi}-\lambda_{n}\pm 2 v=0, \quad \xi_\rho=\pm1. 
\label{a1}
\end{align} 
For fundamental photon hypersurfaces with zero impact parameter $S_0$, the condition $\theta$-TTS must hold \cite{Galtsov:2019fzq}: 
\begin{align}
\lambda_{a}=\lambda_{n}, \quad \xi_\rho=0.
\end{align} 
In this case, the hypersurfaces spatial section itself is closed due to the fact that the causal region coincides with $M$ (if there are no subdomains of causality violation), but has singular points.

The necessary condition for the compatibility of the umbilical equation (\ref{a5}) with the definition of the causative region $\xi^2_\rho \leq1$ is reduced to the fulfillment of either of the following two inequalities (compare with \cite{Galtsov:2019bty})
\begin{align} 
&4v^2\geq(\sigma_{\varphi\varphi}-\sigma_{nn})^2, \label{a39} \\
&(\sigma_{\varphi\varphi}-\sigma_{nn})(\sigma_{\varphi\varphi}-\sigma_{aa})\geq2v^2.
\end{align}
In particular, the first of the conditions corresponds to the photon region in the Kerr metric \cite{Galtsov:2019bty}, and the second in the Zipoy-Voorhees metric \cite{Galtsov:2019fzq}.

Equation (\ref{a38}) allows us to express the principal curvatures of the spatial section $\lambda_{\varphi}$ and $\lambda_{a}$ in terms of the mean curvature $\lambda\equiv Tr(\sigma)/(m-2)$ \cite{Chen} (we consider only the case  $m =4$):
\begin{align} 
&\lambda_{\varphi}=\lambda+\left\{\lambda-\tilde{\lambda}_n\right\}/\Xi_\rho, \\ 
&\lambda_{a}=\lambda-\left\{\lambda-\tilde{\lambda}_n\right\}/\Xi_\rho,\\
&\tilde{\lambda}_n=\lambda_n-2v\xi_\rho, \quad \Xi_\rho=1-2\xi^2_\rho.
\end{align}  
In the case of $\xi_\rho=\pm1/\sqrt{2}$, the denominator of these expressions $\Xi_\rho$ vanishes, however, they remain finite, since the mean curvature in this case is expressed only through the derivatives of the lapse functions:
\begin{align} 
\lambda=\frac{\lambda_{\varphi}+\lambda_{a}}{2}=\lambda_n\mp\sqrt{2}v.
\end{align}
From the Gauss-Codazzi-Ricci equation \cite{Chen} we obtain the relationship between the mean and Gaussian curvature of each section (see for a review \cite{Yoshino1}):
\begin{align} 
&{}^IR=-2G(\hat{r},\hat{r})+2P(\lambda,\xi_\rho)+\left(\frac{2}{\alpha}\right)D^2\alpha+\frac{\varphi^2}{2\alpha^2}\left\{\hat{r}(\omega)^2-D(\omega)^2\right\}, \label{a50}\\
&P(\lambda,\xi_\rho)=-a\lambda^2+2b \lambda+c, \quad a=\frac{1-\Xi^2_\rho}{\Xi^2_\rho}, \quad b=\frac{\tilde{\lambda}_n+\lambda_n\Xi^2_\rho}{\Xi^2_\rho}, \quad
c=-\frac{\tilde{\lambda}^2_n}{\Xi^2_\rho}, 
\end{align}
where $D$ and ${}^IR$ are the covariant derivative and the Ricci scalar at the intersection $I=S\cap N$, and $G$ is the Einstein tensor of  $M$. In this paper, we consider only the case of closed (without boundary) fundamental photon hypersurface with $\rho=0$ from the complete family, which may have or not have  singular points, and all non-singular points $q\in O^S_\rho$. In this case, we can use the original Gauss-Bonnet theorem and prove a simple topological proposition.
 
\begin{theorem}
If at each point $q\in O^S_\rho$ and at the singular points $q$ of the closed convex section $I$ of the 3-dimensional FPH $S_0$ the condition holds
\begin{align} 
\frac{D(\alpha)^2}{\alpha^2}+\frac{\varphi^2}{4\alpha^2}\left\{\hat{r}(\omega)^2-D(\omega)^2\right\}\geq G(\hat{r},\hat{r}),
\label{a70}
\end{align}
then $I$ has topology of a sphere $\mathbb S^2$ (compare with an analogous proposition for transversaly trapping surfaces  \cite{Yoshino1,Yoshino:2019dty}).
\end{theorem}

\begin{proof} 
For proof, we note that at an ordinary point $q\in O^S_\rho$  
\begin{align} 
P(\lambda,0)=4\lambda_n \lambda-\lambda^2_n=\lambda^2_n+2\lambda_\varphi\lambda_a\geq0,
\end{align}
for convex spatial section $I$. For the singular point $p\in \partial  P_0$, the expression $P(\lambda,0)$ is also obviously non-negative, since the umbilical condition is stronger (all principal curvatures are equal). We now integrate the formula
 ($\ref{a50}$) over $I$ 
\begin{align} 
\int_I{}^IRdI\geq\int_I\Big\{-2G(\hat{r},\hat{r})+\frac{2D(\alpha)^2}{\alpha^2}+\frac{\varphi^2}{2\alpha^2}\left\{\hat{r}(\omega)^2-D(\omega)^2\right\}\Big\} dI.
\end{align}
From here, obviously, our statement follows from the Gauss-Bonnet theorem (every closed surface with a positive Euler characteristic has the topology of a sphere).
\end{proof} 

\begin{remark}
If the condition ($\ref{a70}$) is violated, then the solution, generally speaking, may contain fundamental photon hypersurfaces $S_0$ with spatial section of a different topology, for example, toric $\mathbb T^2$. Such a surface can be invariant (axially symmetric), has a zero Euler  characteristic, and, accordingly, Killing fields may not have singular points on it. Arrangement of axially symmetric closed convex surfaces of a different genus $g$ seems difficult.  
\end{remark}

\subsection{Coordinate system choice}
\label{SS4}
A fairly general metric satisfying the orthogonality properties $\LP \hat{n},\hat{\varphi}\RP=0$ is written as \cite{Yoshino1}:
\begin{align}
ds^2=-\alpha^2 d\tau^2+\gamma^2(d\varphi-\omega d\tau)^2+\phi^2 d\eta^2+\psi^2 d\zeta^2,
\label{a2}
\end{align}
where all metric functions are defined on a two-dimensional submanifold with coordinates $\left\{\eta, \zeta\right\}$. In the most general form, an invariant hypersurface in a given coordinate system can be associated with a curve in the subspace $\left\{\eta, \zeta\right\}$:
\begin{align}
\eta=f(s), \quad \zeta=g(s),
\end{align}
where $s$ is an arbitrary real parameter. Components of the second fundamental form and a normal to such a surface are:
\begin{align}
&\left\{\sqrt{\dot{g}^2\psi^2+\phi^2\dot{f}^2}\right\}\hat{r}=\dot{g}(\psi/\phi)\partial_\eta-\dot{f}(\phi/\psi)\partial_\zeta,\\
&\left\{\sqrt{\dot{g}^2\psi^2+\phi^2\dot{f}^2}\right\}\sigma_{\varphi\varphi}=-\dot{g}(\psi/\phi)(\ln\gamma)_\eta+\dot{f}(\phi/\psi)(\ln\gamma)_\zeta,\\
&\left\{\sqrt{\dot{g}^2\psi^2+\phi^2\dot{f}^2}\right\}^3\sigma_{aa}/(\psi\phi)=(\dot{g}\ddot{f}-\dot{f}\ddot{g})+\dot{f}^3(\phi^2/\psi^2)(\ln\phi)_\zeta-\dot{g}^3(\psi^2/\phi^2)(\ln\psi)_\eta\nonumber\\
&+\dot{g}\dot{f}\Big[\left\{2\dot{g}(\ln\phi)_\zeta+\dot{f}(\ln\phi)_\eta\right\}-\left\{2\dot{f}(\ln\psi)_\eta+\dot{g}(\ln\psi)_\zeta\right\}\Big].
\label{a6}
\end{align}
In particular, the first necessary condition (\ref{a39}) for the existence of a fundamental photon region of the Kerr type reads:
\begin{align}
\frac{\hat{r}(\omega)^2 \gamma^2}{\alpha^2}\geq\hat{r}(\ln \gamma/\alpha)^2.
\end{align}

The simplest case is represented by hypersurfaces of the form $\eta={\rm const}$. In this case, it is convenient to choose the natural parameterization $g(s)=s$ (applicable also in other cases), then we obtain:
\begin{align}
\sigma_{\varphi\varphi}=-\left(\frac{1}{\phi}\right)\partial_\eta\ln\gamma, \quad 
\sigma_{\vartheta\vartheta}=-\left(\frac{1}{\phi}\right)\partial_\eta\ln\psi, \quad \hat{r}=\left(\frac{1}{\phi}\right)\partial_\eta.
\end{align}
The umbilical equation (\ref{a5}) and the necessary condition (\ref{a39}) then take the form:
\begin{align}
\frac{\alpha^2\rho^2}{\gamma^2} \left(\partial_\eta\ln\frac{\gamma^2}{\psi^2}\right)+\left(\partial_\eta\ln\frac{\psi^2}{\alpha^2}\right)(1-\omega \rho)^2
=2(1-\omega \rho)\rho\partial_\eta\omega, \quad
\frac{\omega^2_\eta \gamma^2}{\alpha^2}\geq(\ln\gamma/\alpha)^2_\eta
\label{a3}
\end{align}
Note that for coordinates in which $\psi=\phi=\Omega$ (Weyl type \cite{Griffiths}) the equations (\ref{a6}) are simplified:
\begin{align}
&\Omega\left\{\sqrt{\dot{g}^2+\dot{f}^2}\right\}\hat{r}=\dot{g}\partial_\eta-\dot{f}\partial_\zeta,\\
&\Omega\left\{\sqrt{\dot{g}^2+\dot{f}^2}\right\}\sigma_{\varphi\varphi}=\dot{f}(\ln\gamma)_\zeta-\dot{g}(\ln\gamma)_\eta,\\
&\Omega\left\{\sqrt{\dot{g}^2+\dot{f}^2}\right\}\sigma_{\vartheta\vartheta}=(\dot{g}\ddot{f}-\dot{f}\ddot{g})/\left\{\dot{g}^2+\dot{f}^2\right\}
+\dot{f}(\ln\Omega)_\zeta-\dot{g}(\ln\Omega)_\eta.
\end{align}

\setcounter{equation}{0}

\section{Fundamenthal Photon Regions and examples}
\label{S2}

\subsection{Fundamenthal Photon Regions}
\label{SS5}
We now define the concept of a fundamental photon region and a fundamental photon function - a generalization of the classical three-dimensional photon region in the Kerr metric
  \cite{Grenzebach,Grenzebach:2015oea}.

\begin{definition}
The fundamental photon function $PF$ will be called the mapping
\begin{align}
PF:\rho \rightarrow \bigcup_{point}\left\{S_\rho\right\}
\end{align} 
which associates with each $\rho$ one or the union of several FPHs with the same $\rho$.
\end{definition}

\begin{remark}
The function $PF(\rho)$ can be continuous, parametrically defining some connected smooth submanifold in the extended manifold $\left\{M,\rho\right\}$, containing possibly even several different families of Lyapunov periodic orbits \cite{Grover:2017mhm}. At the same time, several continuous functions $PF(\rho)$ can exist in which different FPHs correspond to one $\rho$. In particular, for a given $\rho$, photon  and antiphoton FPHs ((un)stable photon surface \cite{Koga:2019uqd}) can occur simultaneously, indicating the instability of the solution
  \cite{Gibbons}.  
\end{remark}

\begin{definition}
The fundamental photon region is the complete image of the function $PF$
\begin{align}
PR=\bigcup_{\rho} PF(\rho).
\end{align} 
\end{definition}

\begin{remark}
A fundamental photon region is a standard region in the space $M$ in which there are FPOs and, in particular, the classical photon region in the Kerr metric. However, as was noted in \cite{Cederbaum:2019vwv}, this definition can be improved by adding to each point in the PR the subset corresponding to the captured directions in the tangent space. Nevertheless, it is clear that the only essential parameter determining $\hat{\kappa}^{\bot}$ is the parameter $\rho$ of the family, and therefore the choice of $PF(\rho)$ for the analysis of optical properties seems appropriate. The mapping $PF(\rho)$ can several times cover the image of $PR$ or part of it when the parameter $\rho$ is continuously changed. For example, in the case of a static space, $PR$ is covered at least 2 times, i.e. $PF(\rho)$ is a two-sheeted function.
\end{remark}

\subsection{Examples and numerical procedure}
\label{SS6}
As an illustration of the application of the above formalism, we consider its application to three examples of explicit solutions: Kerr, Zipoy-Voorhees with $\delta=2$\cite{Zipoy,Voorhees:1971wh,Griffiths,Kodama:2003ch}, the   $\delta=2$ Tomimatsu-Sato, \cite{Kodama:2003ch}, optical shadows for which were obtained in \cite{Grenzebach, Abdikamalov:2019ztb, Bambi:2010hf}. 

We will compare the structure of $PR$ and $PF(\rho)$ (which is continuous in these cases) of non-extremal solutions $0<p, q<1$ in spheroidal coordinates in an asymptotically flat region for which a metric in the form (\ref{a2}) is
\begin{align}
ds^2=-\alpha^2 d\tau^2+\gamma^2(d\varphi-\omega d\tau)^2+\phi^2 dx^2+\psi^2 dy^2,
\label{a30}
\end{align} 
where $x>1$, $-1\leq y\leq1$, and all metric components depend only on $\left\{x,y\right\}$.

The investigated solutions have additional $\mathbb Z_2$ symmetry under $y$-reflection  relative to the plane $y=0$. Thus, it is convenient to search for fundamental photon hypersurfaces with additional $\mathbb Z_2$ symmetry. Of course, FPHs without such $\mathbb Z_2$ symmetry can also exist (in pairs), as indicated by the existence of $\mathbb Z_2$-asymmetric fundamental photon orbits in some two-center solutions. In the TS/ZV case, they can be located in the vicinity of two horizons and be essentially non-spherical in the coordinates $(x,y)$. To find them, you can use the coordinates of Kodama and Hikida $(X,Y)$.

\begin{figure*}[tb]
\centering
\subfloat[][$\rho_{min}<\rho<\rho_{max}$]{
  	\includegraphics[scale=0.35]{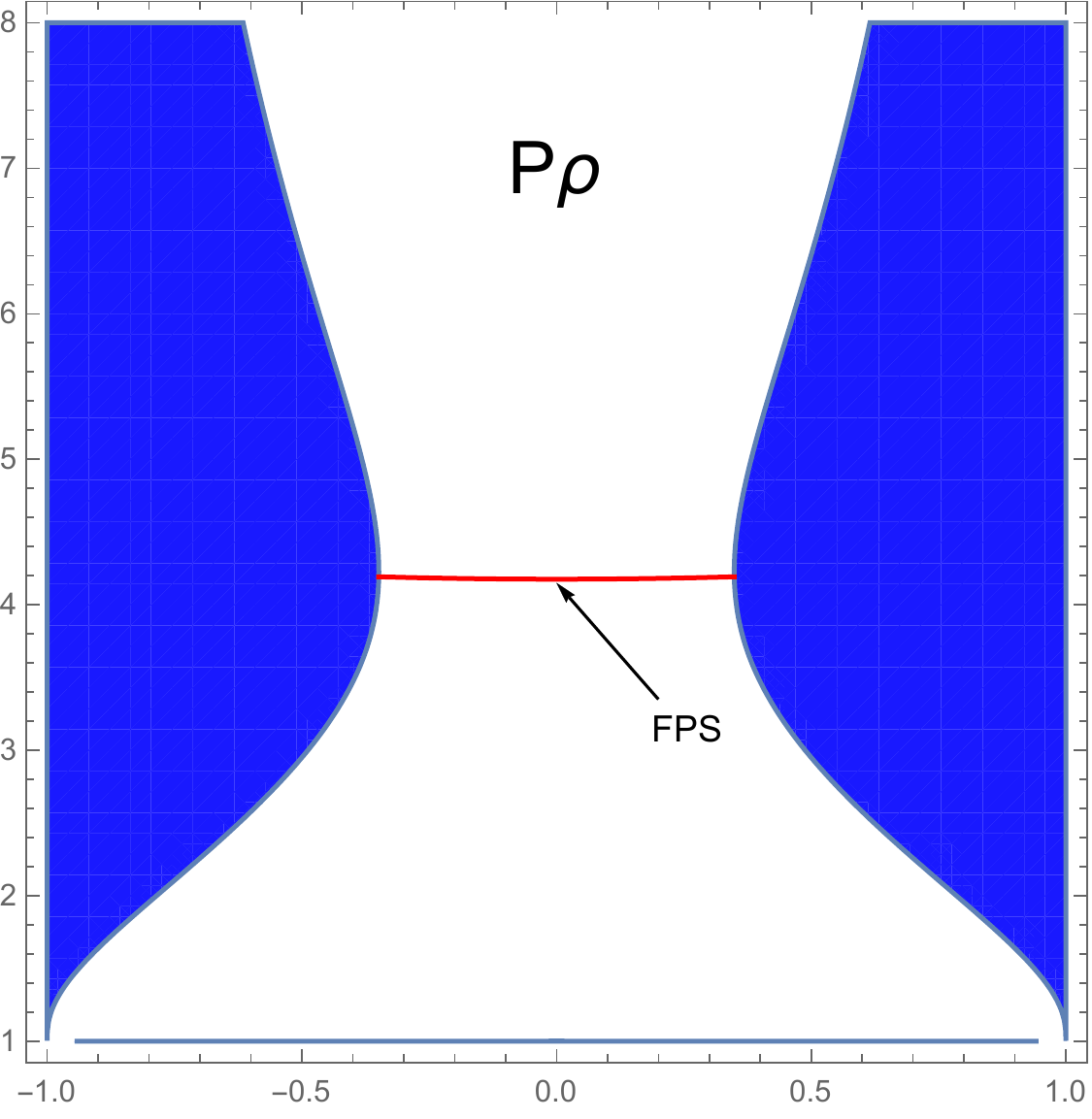} 
		\label{PR1}
 }
\subfloat[][$\rho=\rho_{max}$]{
  	\includegraphics[scale=0.35]{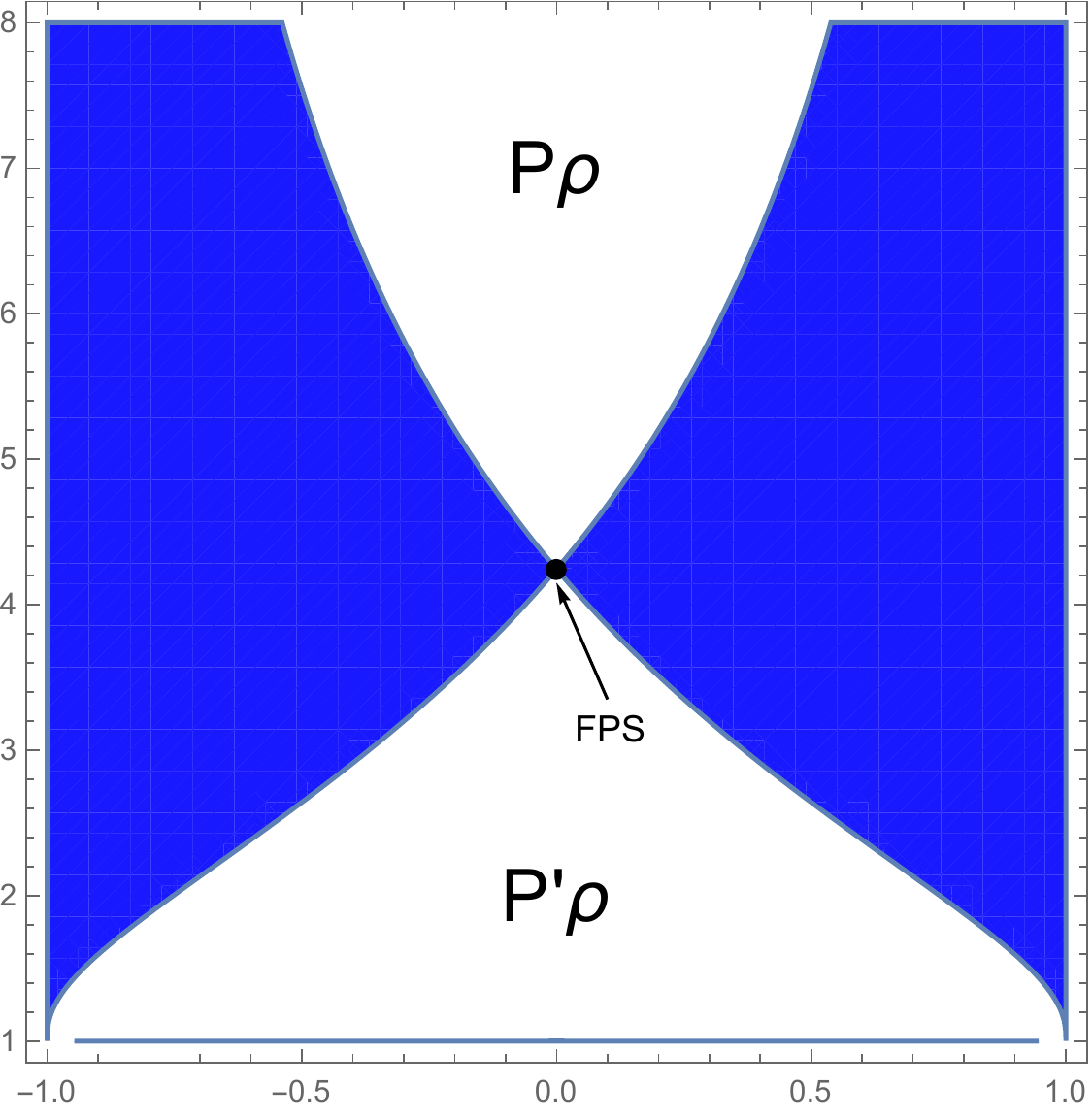} 
		\label{PR2}
 }
\subfloat[][$\rho>\rho_{max}$]{
  	\includegraphics[scale=0.35]{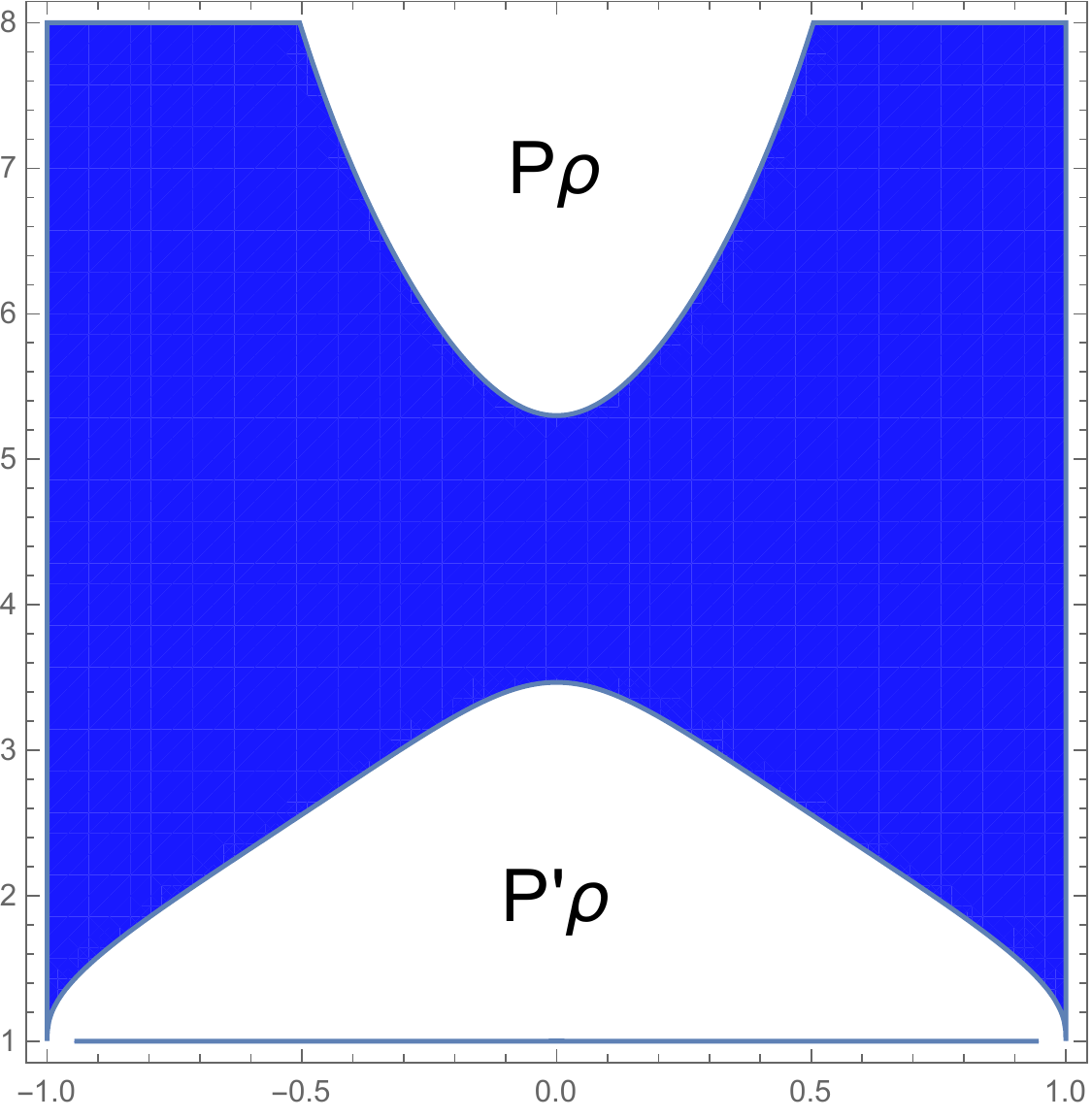} 
		\label{PR3}
 }
\caption{Causal region in the Tomimatso-Sato solution.}.
\label{P}
\end{figure*}
In all the examples we are considering, the causal region $P_\rho$ (the accessible region of some effective potential) contains both horizons/singularity and spatial infinity if and only if $\rho_{min}<\rho<\rho_{max}$ Fig. \ref{PR1}. Otherwise, there are two connected regions $P_\rho$, one of which contains spatial infinity, and the other --- the horizon/singularity Fig. \ref{PR3}. Thus, spatial infinity is separated from the horizon/singularity (no null geodesics with this impact parameter can  connect the horizon and spatial infinity). Therefore, we will consider only the range $\rho_{min}<\rho<\rho_{max}$ as the domain of definition of the function $PF(\rho)$.

To determine the values of the parameter $\rho$ at which the causal region breaks, we use the $\mathbb Z_2$ symmetry of the solution and find the conditions under which this discontinuity occurs at the equatorial plane:
 \begin{align}
&(1-\omega \rho)^2\varphi^2-\alpha^2\rho^2|_{y=0}=0,\\
&\partial_y\left\{(1-\omega \rho)^2\varphi^2-\alpha^2\rho^2\right\}|_{y=0}=0,
\end{align}
where the first condition means that the boundary of the causal region intersects the equatorial plane, and the second, that for larger and smaller $x$ we again fall into the causal region, that is, the desired point is really a discontinuity point Fig. \ref{PR2}. It is easy to verify that we obtain exactly the familiar conditions (\ref{a1}), but limited to the equatorial plane.
\begin{align}
&\lambda_{\varphi}-\lambda_{n}\pm2 v=0|_{(y=0,f'(0)=0)}, \\ &\rho_\pm=\frac{\gamma}{\gamma\omega\pm\alpha}|_{y=0}.
\label{a8}
\end{align}  
As a result, there will be equatorial circular photon orbits (fundamental photon submanifolds of dimension $n=2$ or $\varphi$-TTS) at the discontinuity points, and the fundamental photon region will interpolate between them, similar to what it was in Zipoe-Voorhees \cite{Galtsov:2019fzq}.

To find $\mathbb Z_2$ - symmetric fundamental photon hypersurfaces, we use the shooting method described in \cite{Galtsov:2019fzq}. We solve the differential equation (\ref{a5}) numerically with a choice of parameterization in (\ref{a6}) of the form $g(y)=y$. Moreover, the surface is uniquely determined by the function $f(y)$. The initial conditions are of the form $f(0)=x_\rho$ and $\dot{f}(0)=0$, where $x_\rho$ is determined by the condition that the boundary points of the resulting hypersurfaces satisfy the boundary condition (\ref{a1}).
\begin{figure*}[tb]
\centering
\subfloat[][$J=0.1$]{
  	\includegraphics[scale=0.35]{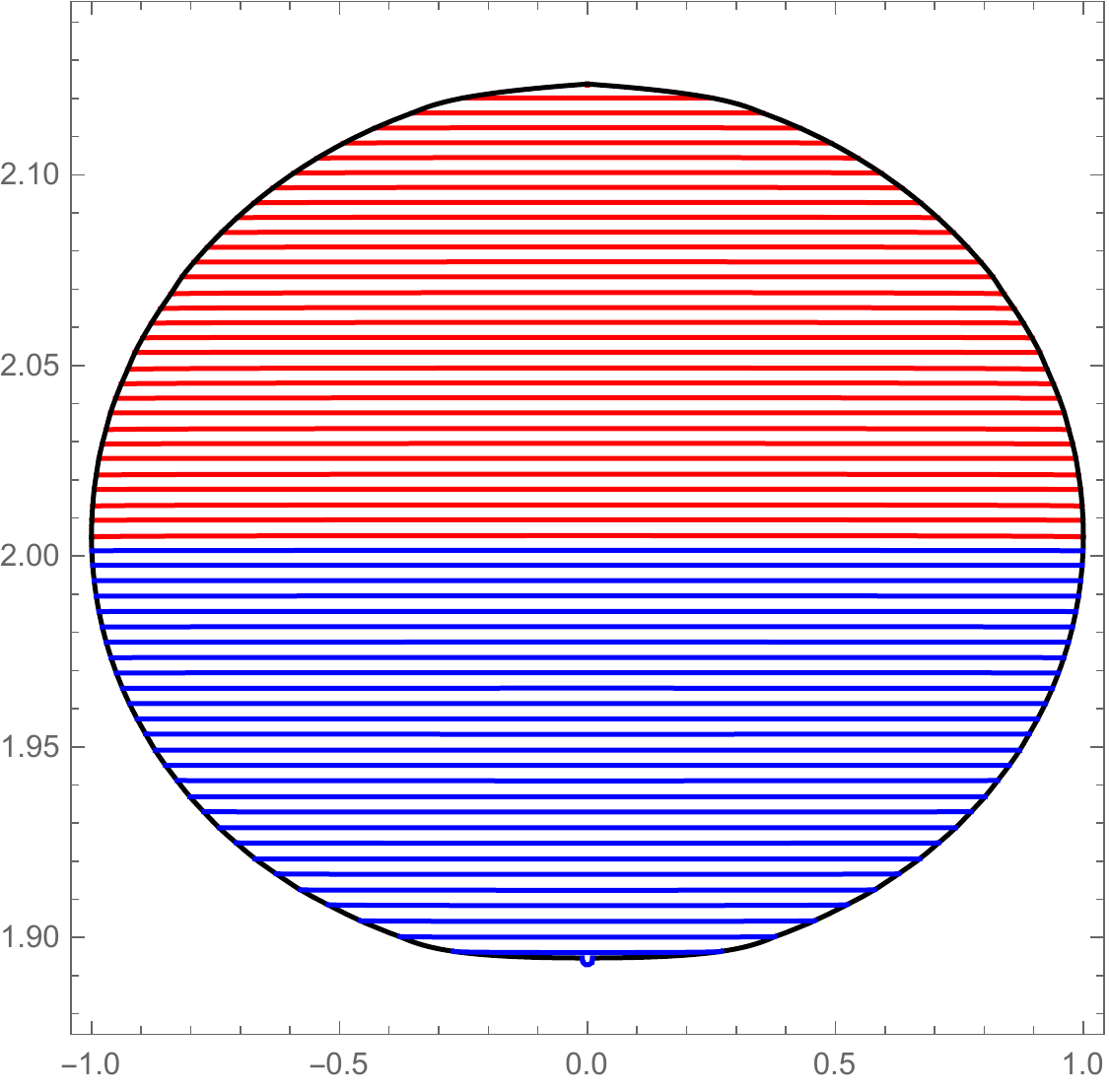} 
		\label{KPR1}
 }
\subfloat[][$J=0.5$]{
  	\includegraphics[scale=0.35]{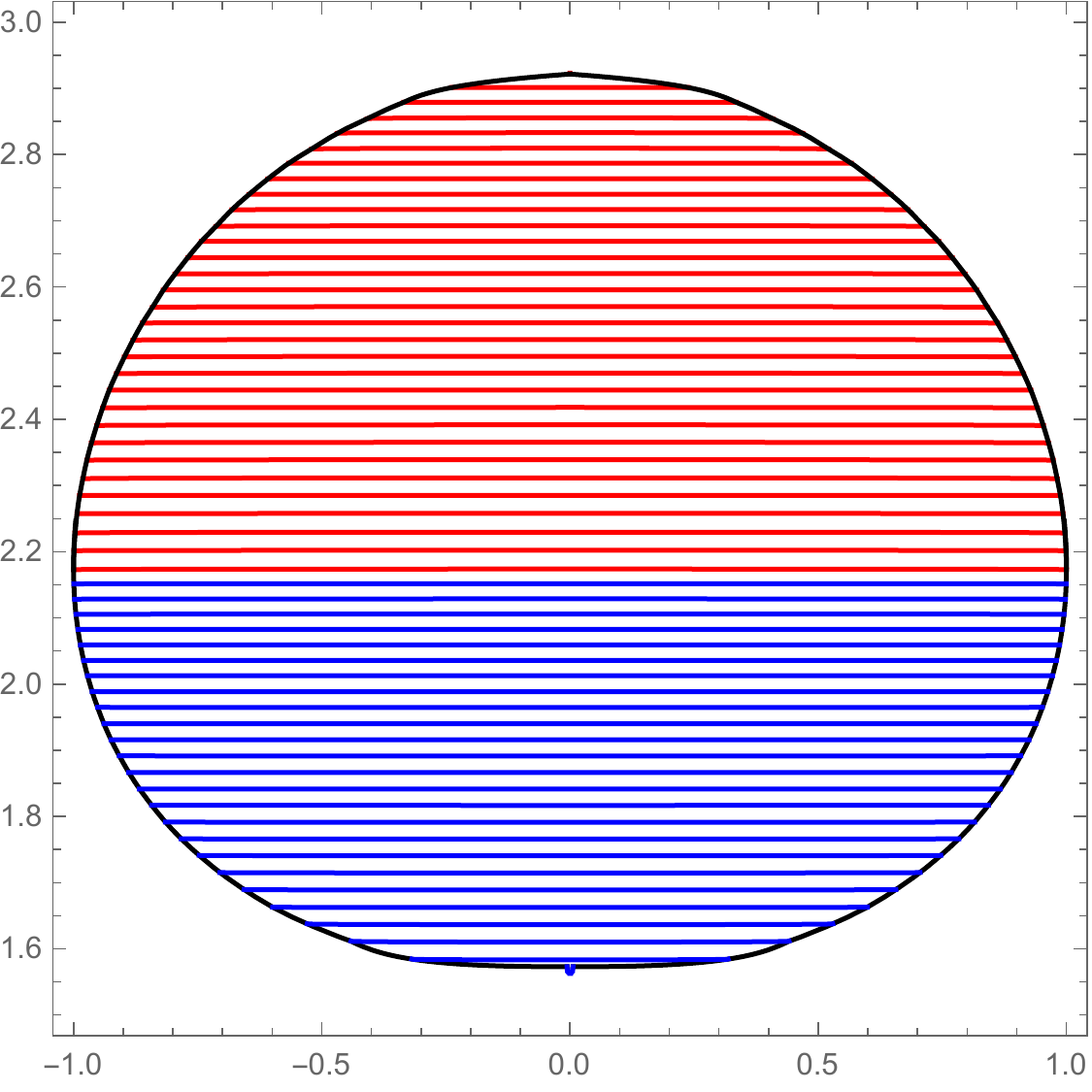} 
		\label{KPR5}
 }
\subfloat[][$J=0.9$]{
  	\includegraphics[scale=0.35]{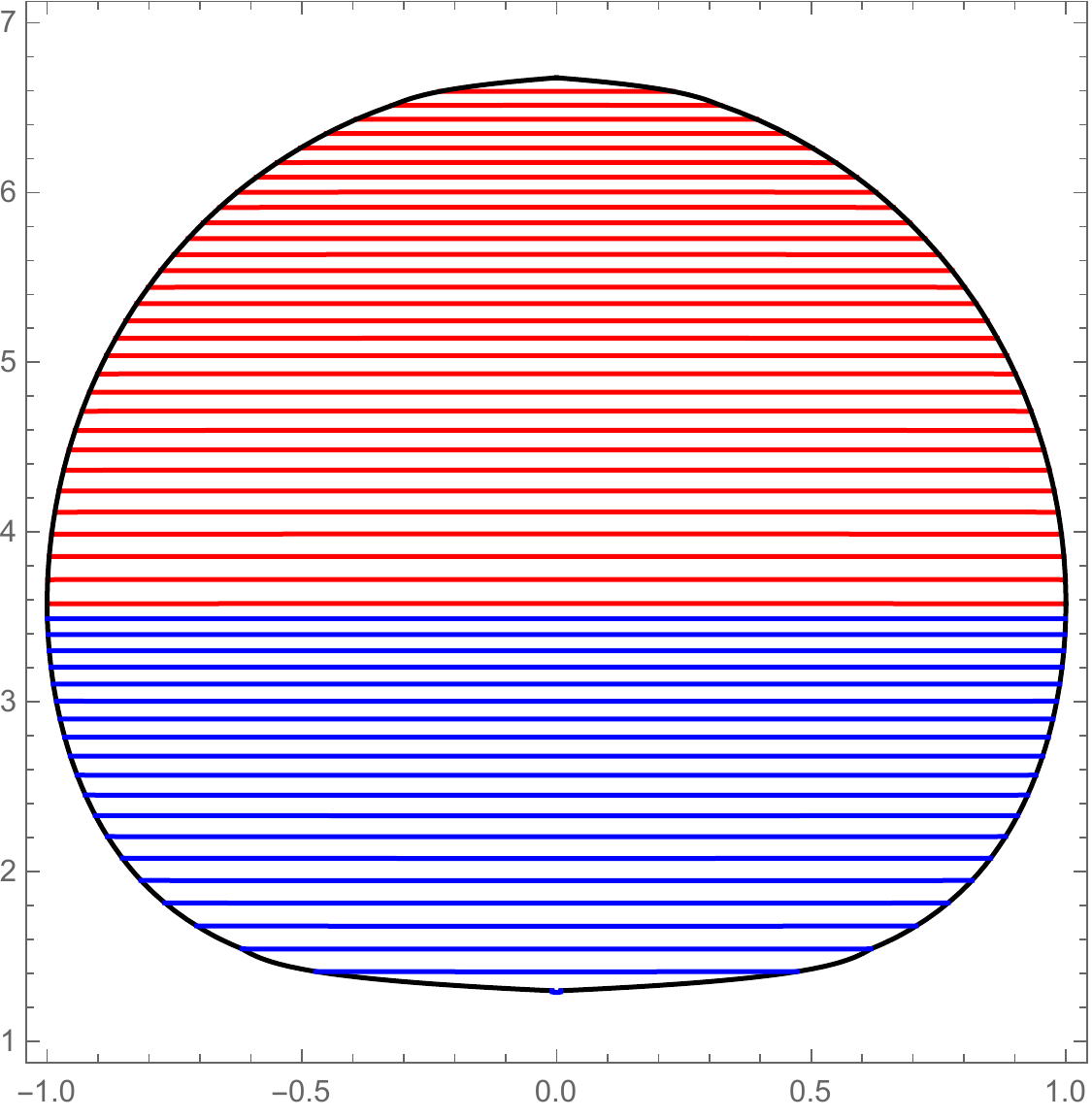} 
		\label{KPR5}
 }\\
\subfloat[][$J=0.1;0.5;0.9$]{
  	\includegraphics[scale=0.35]{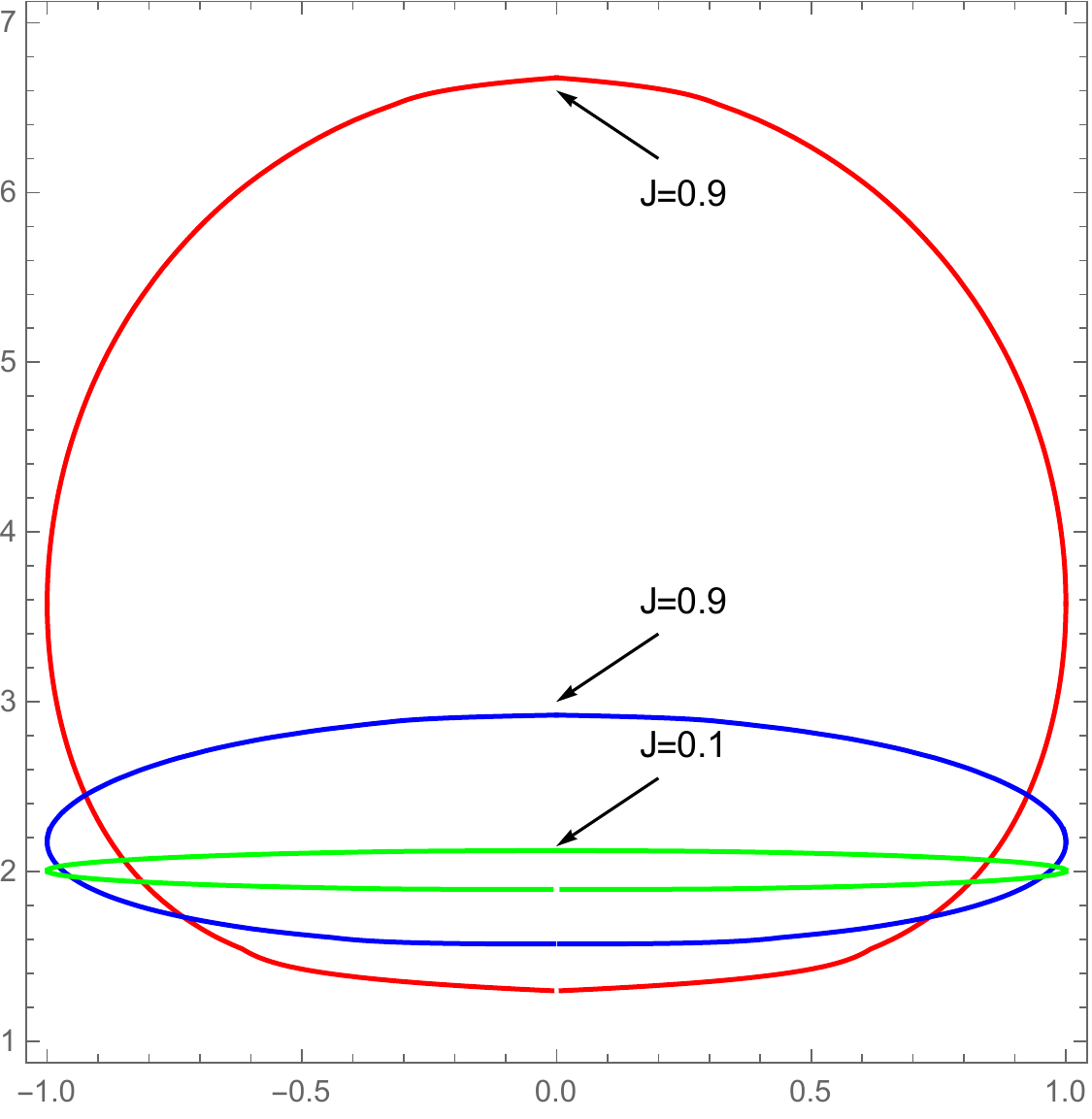} 
		\label{KPRK}
 }
\caption{Photon region in the Kerr metric.}
\label{K}
\end{figure*}

\subsection{Kerr metric}
\label{SS7}
Now let's analyze a concrete example of a Kerr solution using spheroidal coordinates. The metric has the form (\ref{a30}) with the following metric functions (see \cite{Griffiths}):
\begin{align}
&\alpha^2=\frac{\tilde{\gamma}^2\tilde{\alpha}^2}{\tilde{\gamma}^2-\tilde{\alpha}^2\tilde{\omega}^2}, \quad \gamma^2=\tilde{\gamma}^2-\tilde{\alpha}^2\tilde{\omega}^2, \quad \omega=\frac{\tilde{\omega}\tilde{\alpha}^2}{\tilde{\gamma}^2-\tilde{\alpha}^2\tilde{\omega}^2},\\
&\phi^2=m^2\frac{(p x+1)^2+q^2y^2}{x^2-1}, \quad
\psi^2=m^2\frac{(p x+1)^2+q^2y^2}{1-y^2}.
\end{align}
where
\begin{align}
\tilde{\alpha}^2=\frac{p^2x^2+q^2y^2-1}{(p x+1)^2+q^2y^2}, \quad \tilde{\gamma}^2=\tilde{\alpha}^{-2}(x^2-1)(1-y^2), \quad \tilde{\omega}=-\frac{2q(1-y^2)(p x+1)}{p(p^2x^2+q^2y^2-1)}.
\end{align}
Here $q$ is the rotation parameter associated with the angular momentum $J=M^2q$, where $M=m$ is the ADM mass of the solution and $p=\sqrt{1-q^2}$. In the future, we will compare solutions that have the same physical parameters $M$ and $J$.

In the Kerr metric, the causal region $P_\rho$ contains both horizon/singularity and spatial infinity if and only if $p=\sqrt{1-q^2}$, where the minimum and maximum values of $\rho$ are determined from (\ref{a8}) as the maximum and minimum roots of the equation
\begin{align}
q\left\{3+(q-\rho/m)^{2/3}\right\}+\left\{-3+(q-\rho/m)^{2/3}\right\}\rho/m=0.
\end{align}
and correspond to two equatorial circular photon orbits. Apart from them, as is well known in the Kerr metric, so-called spherical orbits exist with constant value of the Boyer-Lindquist radial coordinate $r$  \cite{Wilkins:1972rs,Teo,Paganini:2016pct}, they correspond to a discrete set of tangential directions on the sphere $r={\rm const} $. Spherical orbits with different $r$ then fill the three-dimensional domain - the {\em photon region} (PR) \cite{Grenzebach,Grenzebach:2015oea,Galtsov:2019bty} which is an important feature of rotating spacetimes. This photon region is a special case of the fundamental photon region introduced by us.
 
In a spheroidal coordinate system, $PR$ and $PF(\rho)$ can be described graphically by considering their section with the plane $\tau={\rm const}$ and $\varphi={\rm const}$ in the adapted coordinate system \cite{Grenzebach}. Moreover, $PR$ is a two-dimensional region on the submanifold - $\left\{x,y\right\}$ as shown in the Figs. \ref{K}, and $PF(\rho)$ is some 2-dimensional submanifold of the three-dimensional space $\left\{x,y,\rho\right\}$. Red and blue lines depict the cross section of individual fundamental photon hypersurfaces $PF(\rho_i)$ with positive and negative values of the impact parameter, respectively. Moreover, for the Kerr metric, the function $PF(\rho_i)$ once covers $PR$ with a continuous change in the parameter of the family $\rho$ i.e. is univalent, and each individual fundamental photon hypersurfaces has the form $x={\rm const}$.

In Kerr metric, such a univalence of $PF(\rho)$ means that the minimum and maximum values of the impact parameter corresponding to the minimum and maximum radii of the equatorial photon orbits and, as a consequence, the minimum and maximum size of the shadow (from the center point to the boundary for the equatorial observer\cite{Abdikamalov:2019ztb,Galtsov:2019fzq}). Thus, the univalent function $PF(\rho)$ corresponds to a shadow with maximum and minimum size at the equatorial plane. The sphericity of the fundamental orbits corresponds to the integrability of the corresponding dynamical system \cite{Pappas:2018opz, Glampedakis:2018blj} and the existence of an additional conserved quantity associated with the Killing tensor \cite{Kubiznak:2007kh}.
\begin{figure}[tb]
\centering
\subfloat[][$PR$]{
  	\includegraphics[scale=0.35]{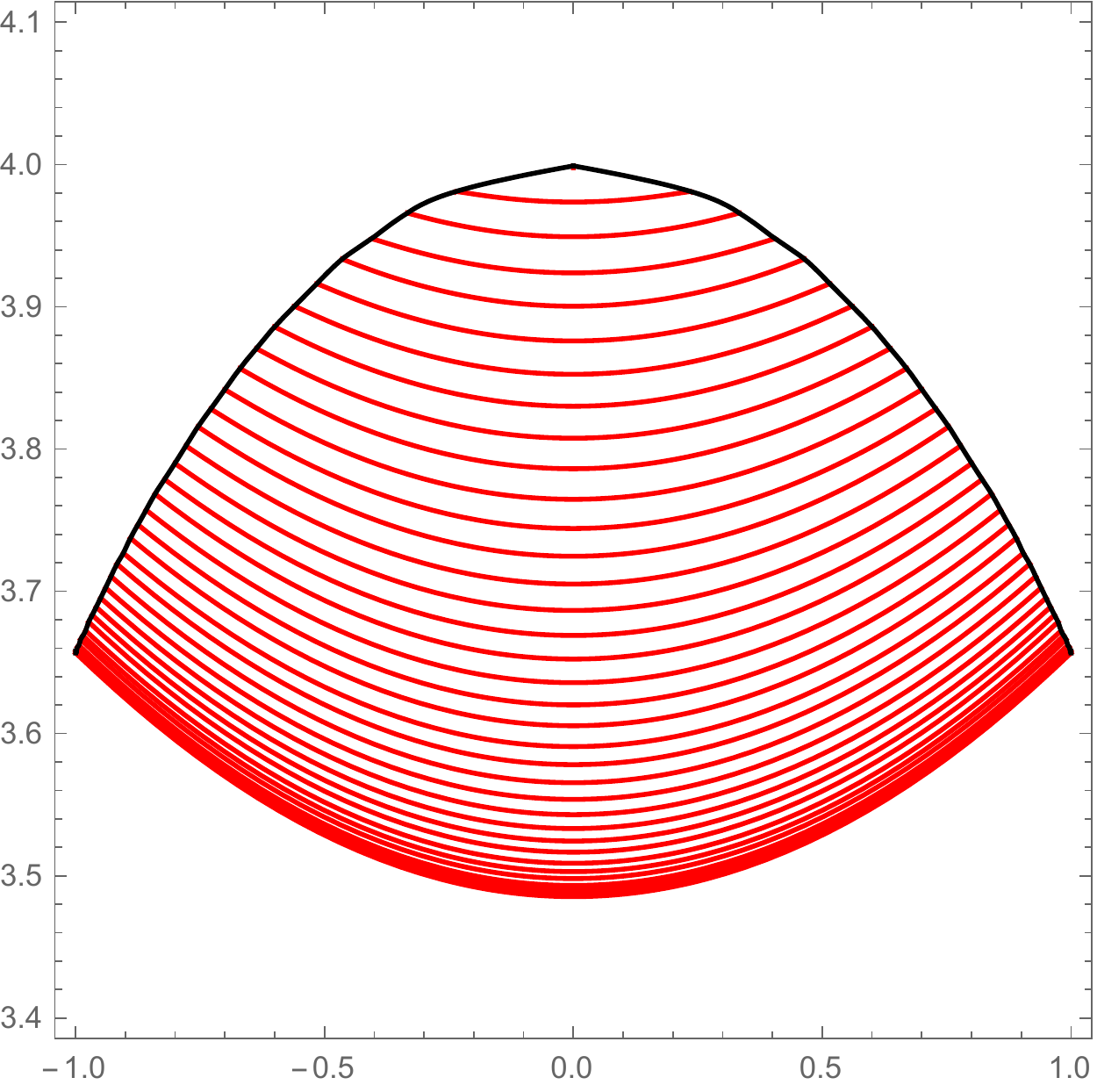} 
		\label{KPZ}
 }
\subfloat[][$PF(\rho)$]{
  	\includegraphics[scale=0.35]{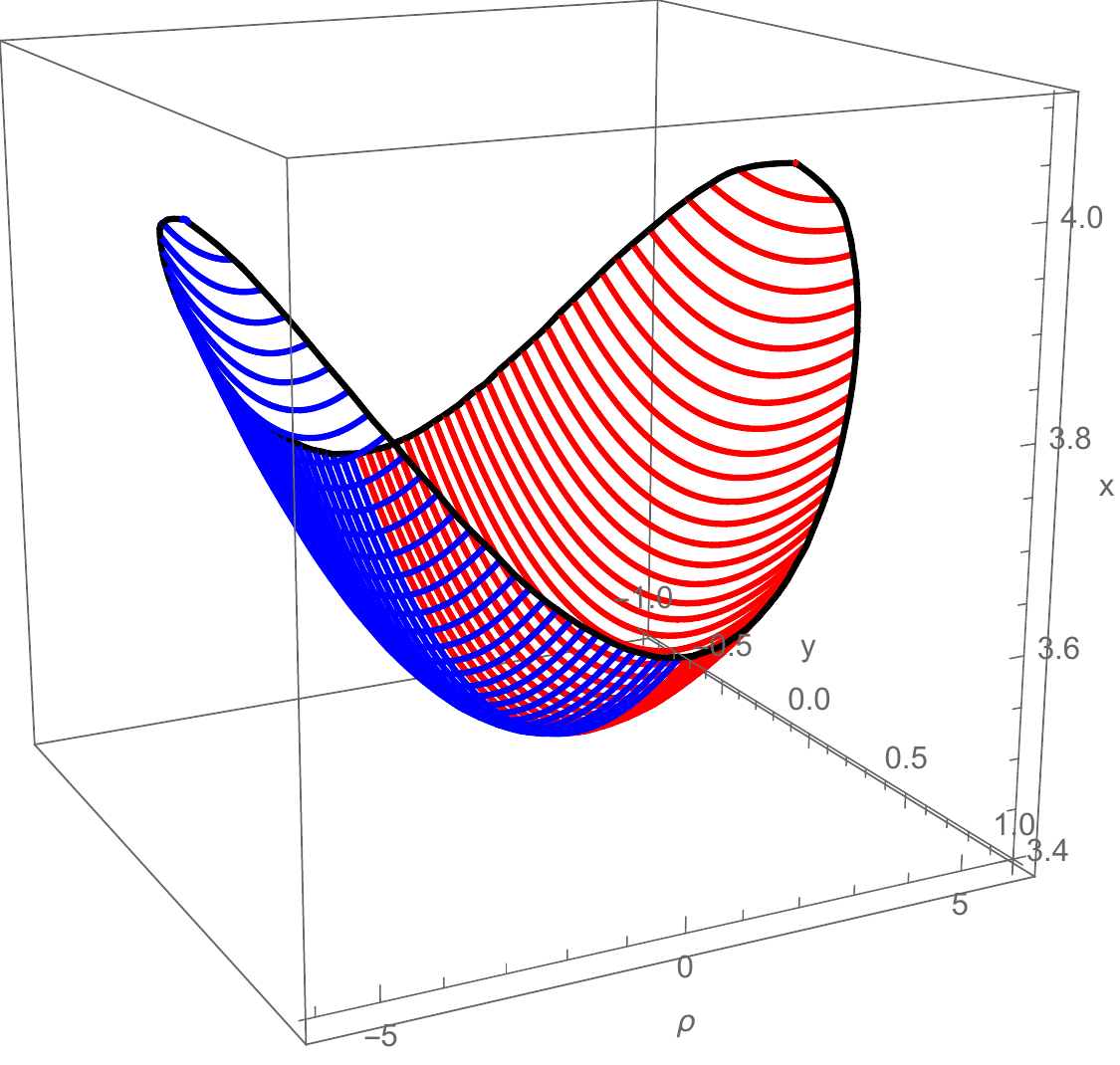} 
		\label{KPZ3D}
 }
\caption{$PR$ and $PF(\rho)$ in the Zipoy-Voorhees metric ($M=1$, $\delta=2$).} 
\label{Z}
\end{figure}

\subsection{Zipoy-Voorhees metric} 
\label{SS8}
As a non-trivial example of a static axially symmetric asymptotically flat spacetime not admitting the standard photon surfaces but contains a non-spherical photon region \cite{Galtsov:2019fzq}, we will consider the Zipoy-Voorhees (ZV) vacuum solution \cite{Zipoy,Voorhees:1971wh,Griffiths,Kodama:2003ch} which in the spheroidal coordinates reads:
\begin{align}
&\alpha^2=\left(\frac{x-1}{x+1}\right)^{\delta}, \quad \gamma^2=m^2\left(\frac{x+1}{x-1}\right)^{\delta}(x^2-1)(1-y^2),\\
&(x^2-1)\phi^2=m^2(x^2-1)^{\delta^2}\left(\frac{x+1}{x-1}\right)^{\delta}(x^2-y^2)^{1-\delta^2},\\
&(1-y^2)\psi^2=m^2(x^2-1)^{\delta^2}\left(\frac{x+1}{x-1}\right)^{\delta}(x^2-y^2)^{1-\delta^2}.
\end{align}
This solution can be interpreted as an axially symmetric deformation of the Schwarzschild metric with the deformation parameter $\delta\geq0$, to which it reduces for $\delta=1 $. For $\delta=2$ it can be interpreted as a two-center solution, a particular non-rotation version of the Tomimatsu-Sato metric \cite{Kodama:2003ch}. The Arnowitt-Deser-Misner mass is equal to  $M=m\delta$. The outer domain in which we are interested in extends as $x>1$.

As in Kerr, the causal region $P_\rho$  contains both singularity and spatial infinity if and only if $\rho_{min}<\rho<\rho_{max}$, where from (\ref{a8}) it can be found
\begin{align}
-\rho_{\rm min}=\rho_{\rm max}=m(2\delta-1)^{-\delta+1/2}(2\delta+1)^{\delta+1/2}.
\end{align}
Otherwise, there are two connected domains $P_\rho$ one of which contains spatial infinity.

In our paper \cite{Galtsov:2019fzq} it was demonstrated that the hypersurfaces of the fundamental photon region (generalized photon region) are significantly different from surfaces of constant radius $x={\rm const}$. As a result, the corresponding dynamical system may contain chaos regions \cite{Lukes} since it contains non-equatorial non-spherical closed photon orbits \cite{Pappas:2018opz, Glampedakis:2018blj}.

The result of the numerical calculation for $\delta=2$ is shown in the Figs. \ref{KPZ}. For an arbitrary $\delta$, the analysis was carried out in detail earlier, however, in a different coordinate system, but the basic laws will obviously be valid here too. Note that the hypersurfaces of the photon region are compressed on the equatorial plane $y=0$ and extend to the poles $y=\pm1$ for $\delta>1$, in addition, the photon hypersurfaces are determined for each possible value of the impact parameter and therefore the solution creates a complete set relativistic images \cite{Virbhadra:1999nm, Virbhadra:2008ws, Virbhadra:2002ju} along the entire border of the shadow.
  
As we said in the case of a static space, $PF(\rho)$ is at least a two-sheeted function. In this case, the fundamental photon region is more appropriate to consider as a hypersurface in the coordinates $\left\{x,y,\rho\right\}$. The corresponding image of the continuous function $PF(\rho)$ in our case is shown in the Fig. \ref {KPZ3D}. The fact that $PF(\rho)$ is two-sheeted leads to the obvious additional symmetry of the shadow of any static axially symmetric solution. In addition, the shadow will have a maximum size along the equatorial section (a more distant fundamental photon surface at the maximum ($\rho$), and the minimum along the vertical, so the shadow of the solution will be flattened in the vertical direction \cite{Abdikamalov:2019ztb,Galtsov:2019fzq}. 
 
\begin{figure*}[tb]
\centering
\subfloat[][$J=0.1$, $PR$]{
  	\includegraphics[scale=0.32]{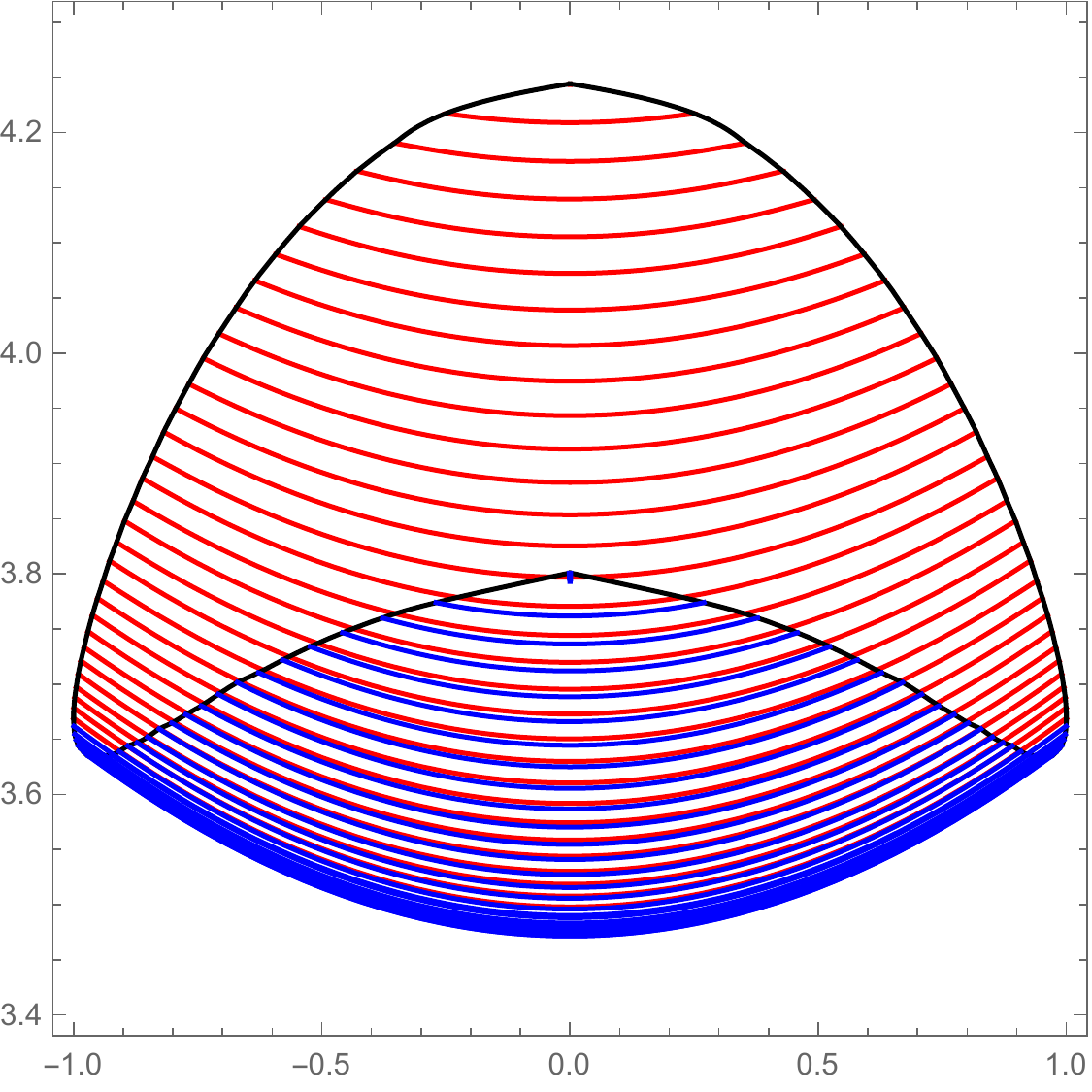} 
		\label{PRTS1}
 }
\subfloat[][$J=0.5$, $PR$]{
  	\includegraphics[scale=0.32]{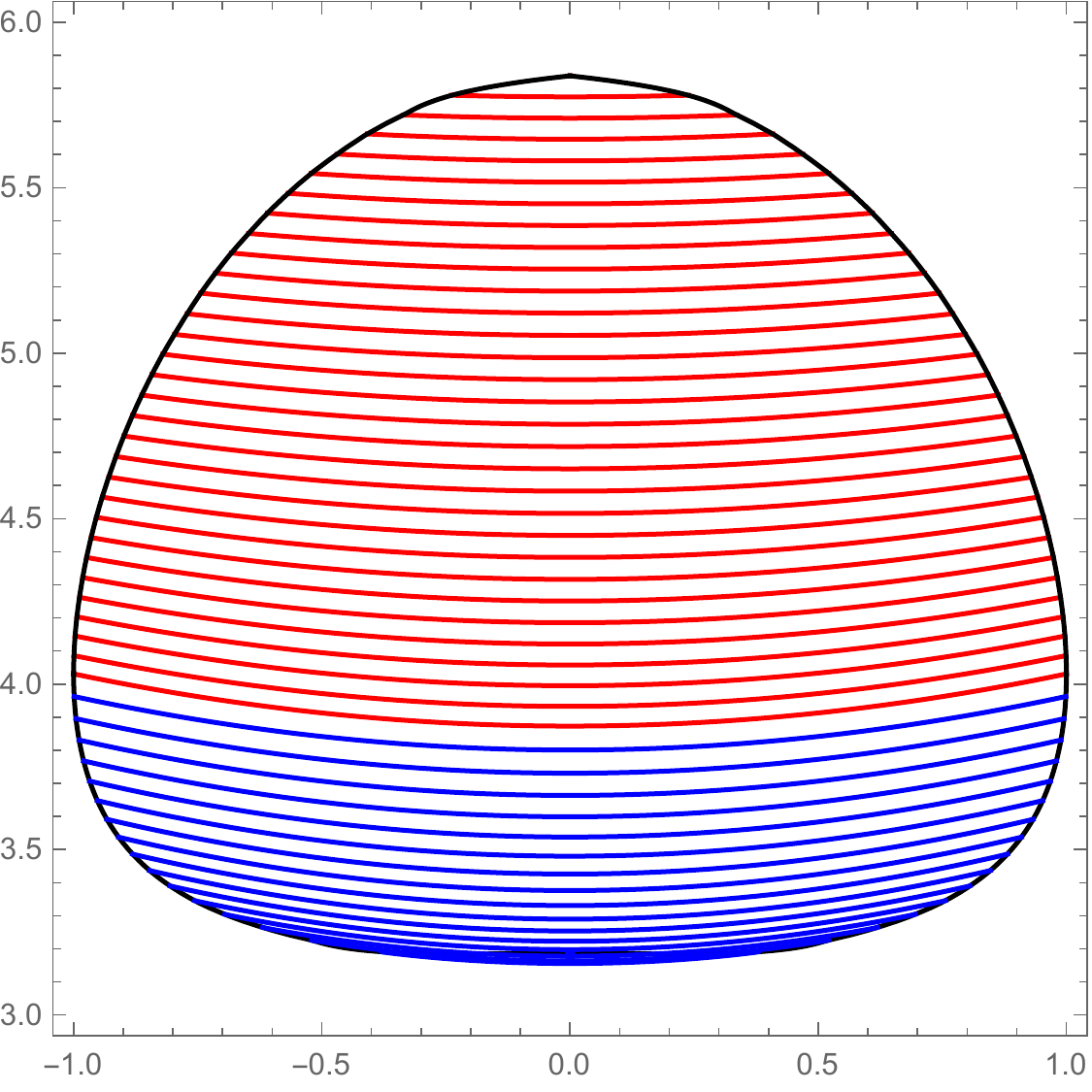} 
		\label{PRZ1}
 } 
\subfloat[][$J=0.9$, $PR$]{
  	\includegraphics[scale=0.32]{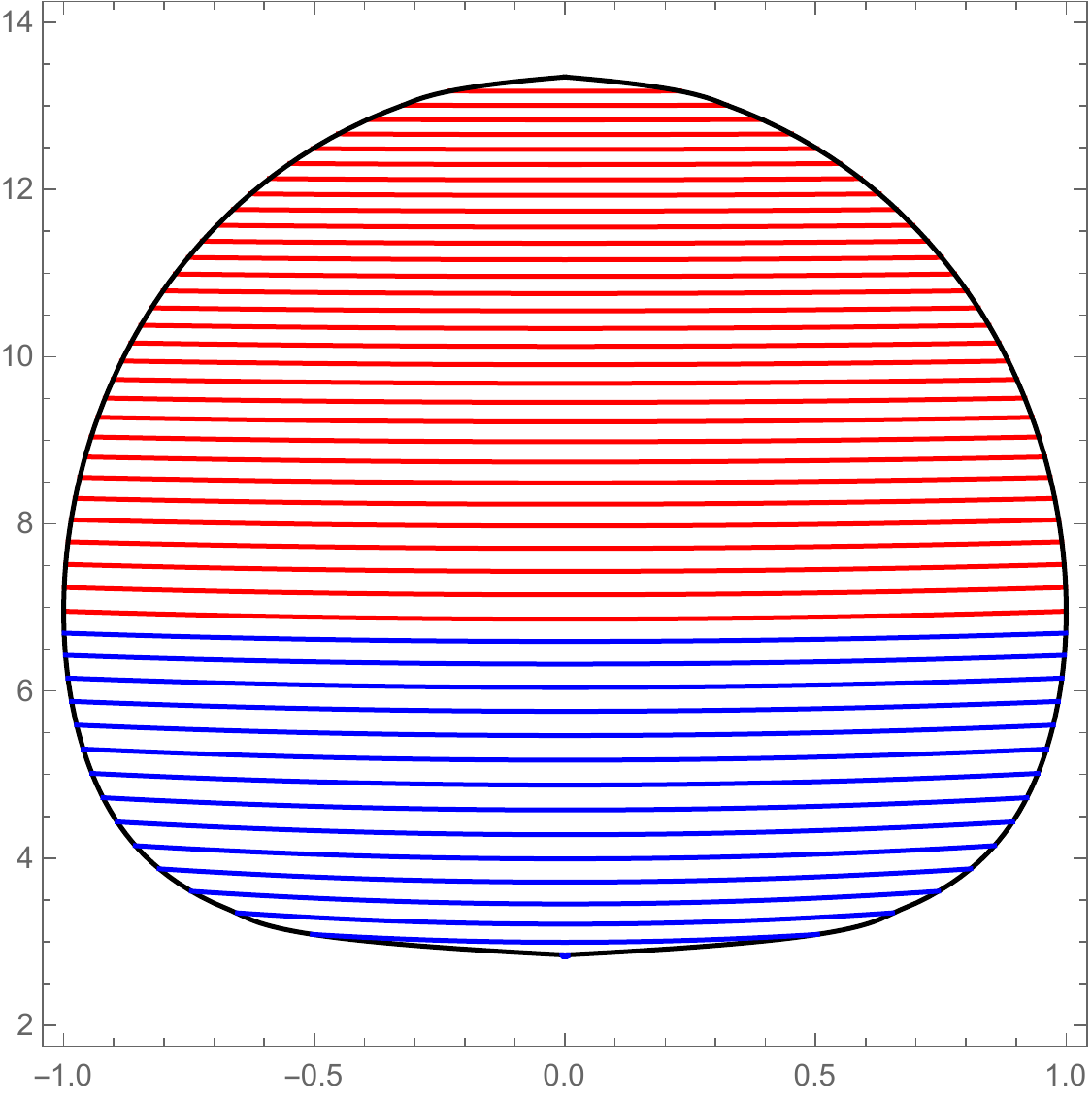} 
		\label{PRZ2}
 }
\\
\subfloat[][$J=0.1$, $PF(\rho)$]{
  	\includegraphics[scale=0.32]{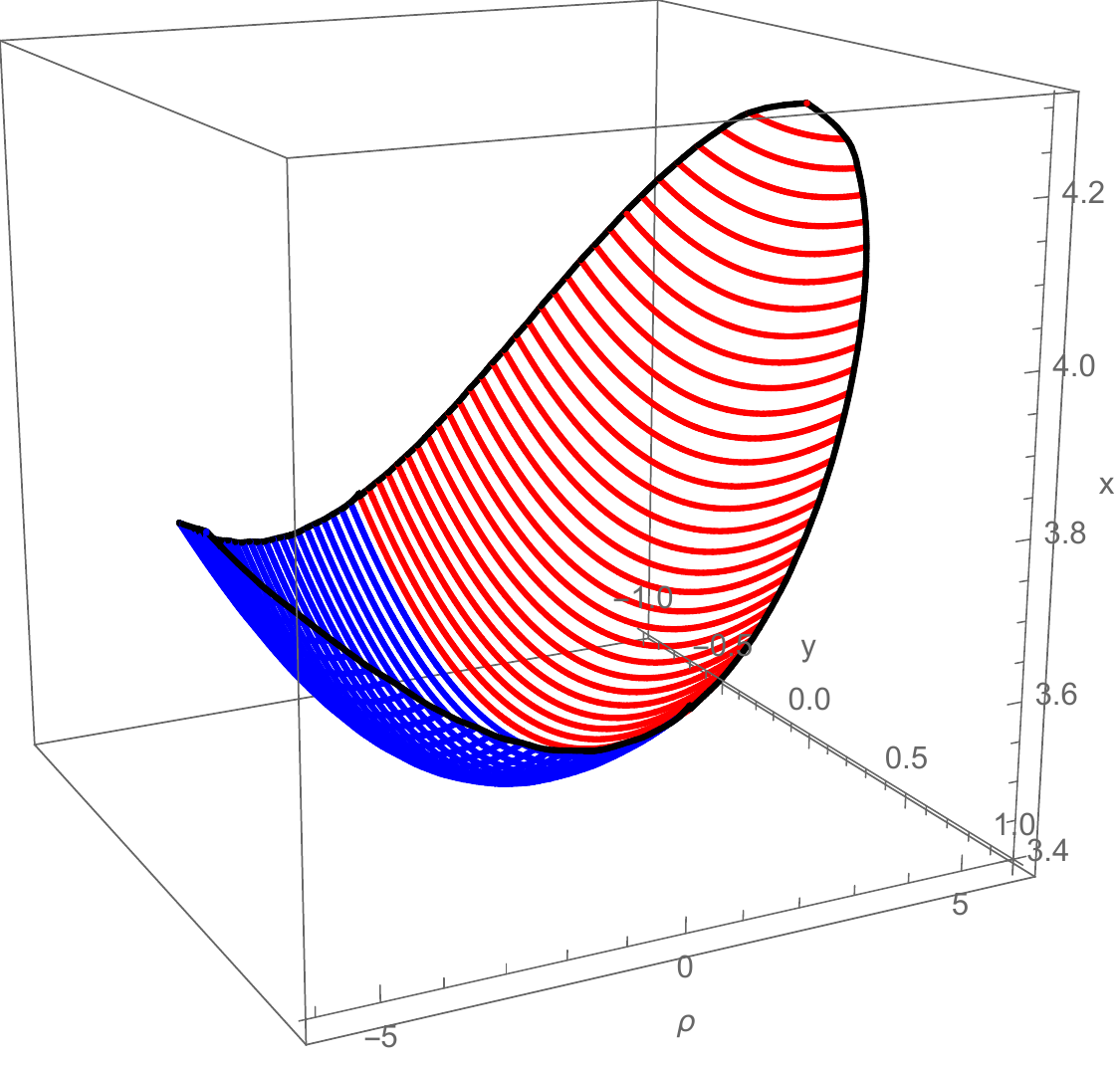} 
		\label{PRTS3D}
 } 
\subfloat[][$J=0.5$, $PF(\rho)$]{
  	\includegraphics[scale=0.32]{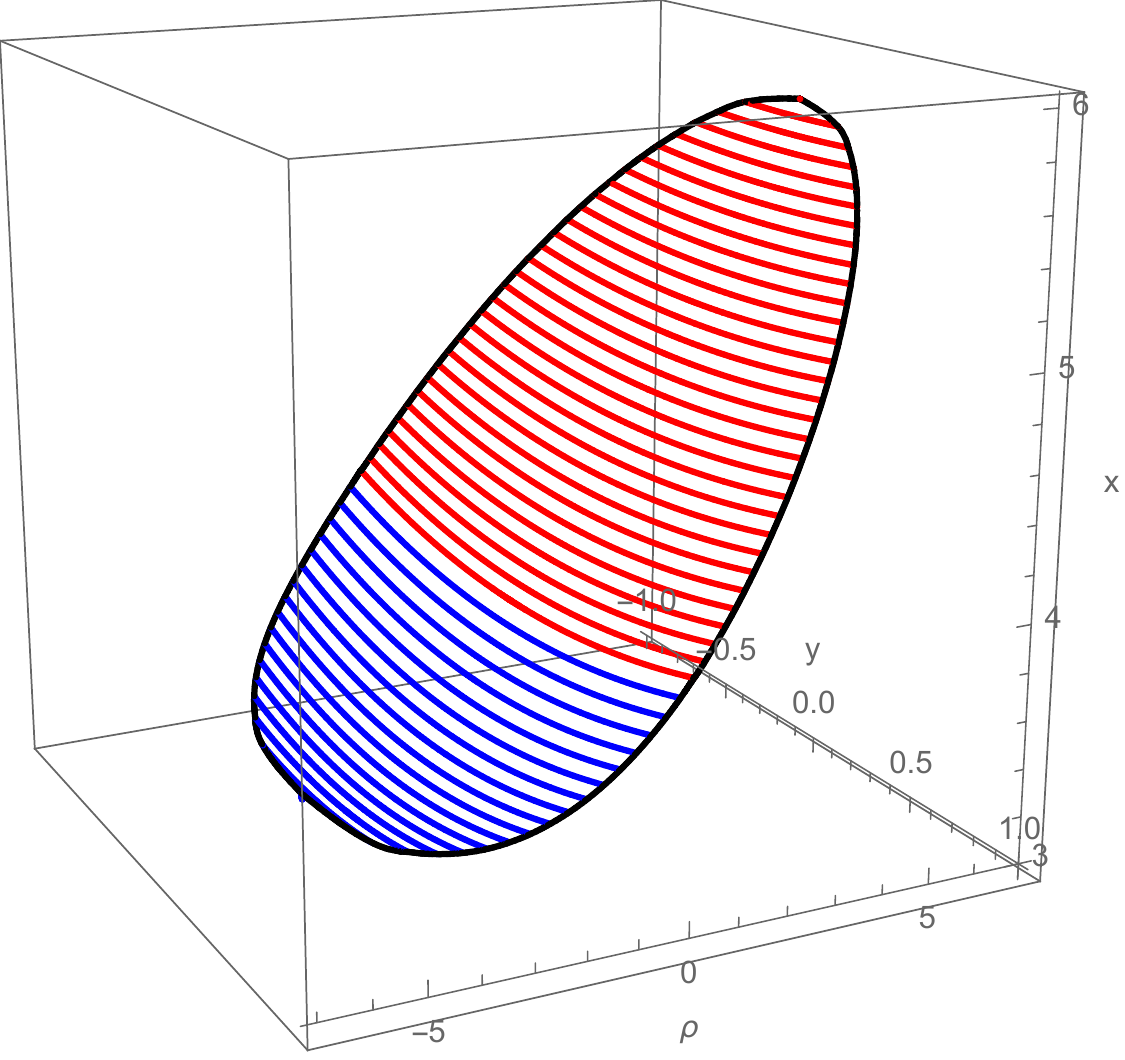} 
		\label{PRZ13D}
 } 
\subfloat[][$J=0.9$, $PF(\rho)$]{
  	\includegraphics[scale=0.32]{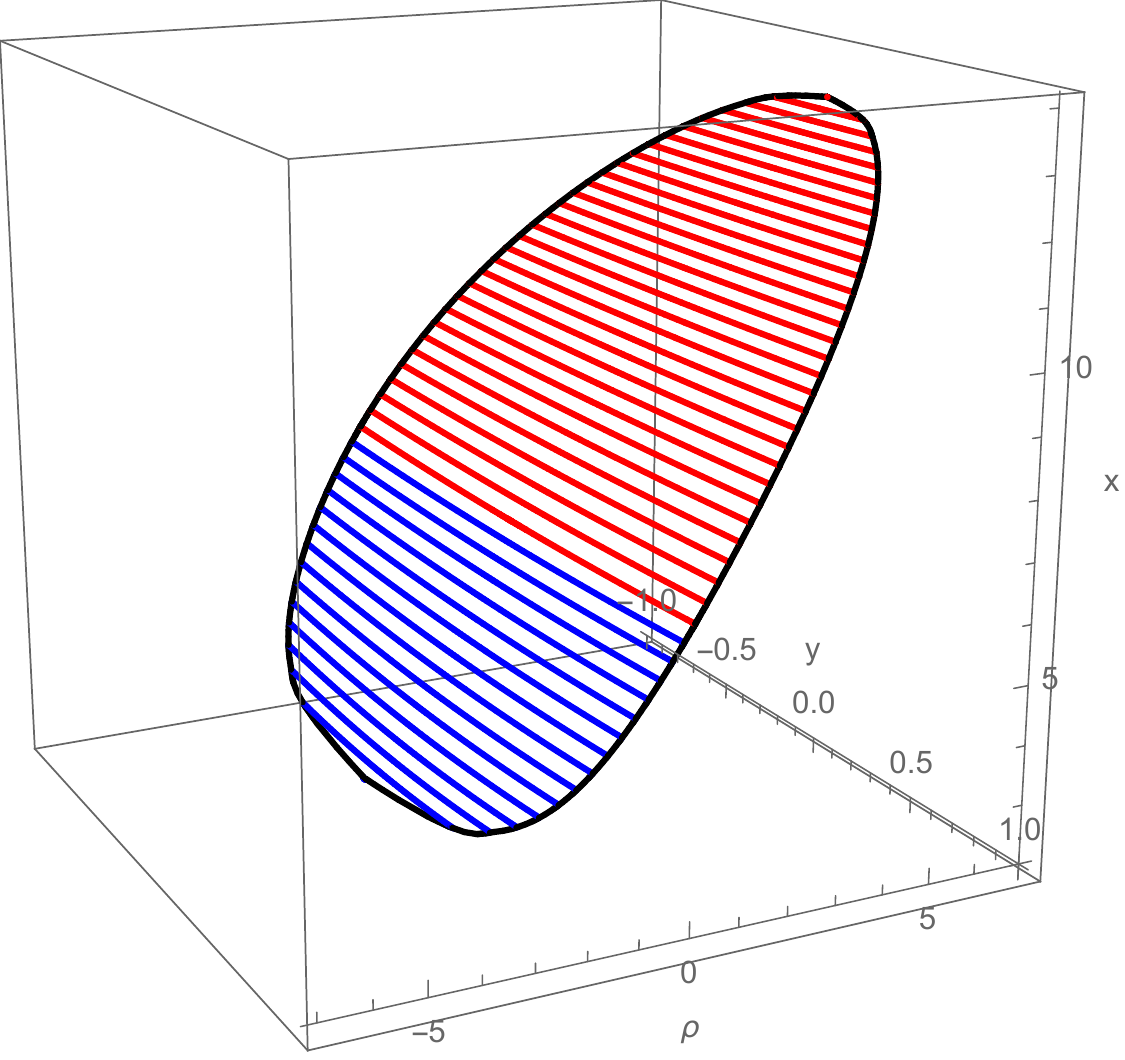} 
		\label{PRZ23D}
 }
\\
\subfloat[][$J=0.1;0.5;0.9$]{
  	\includegraphics[scale=0.32]{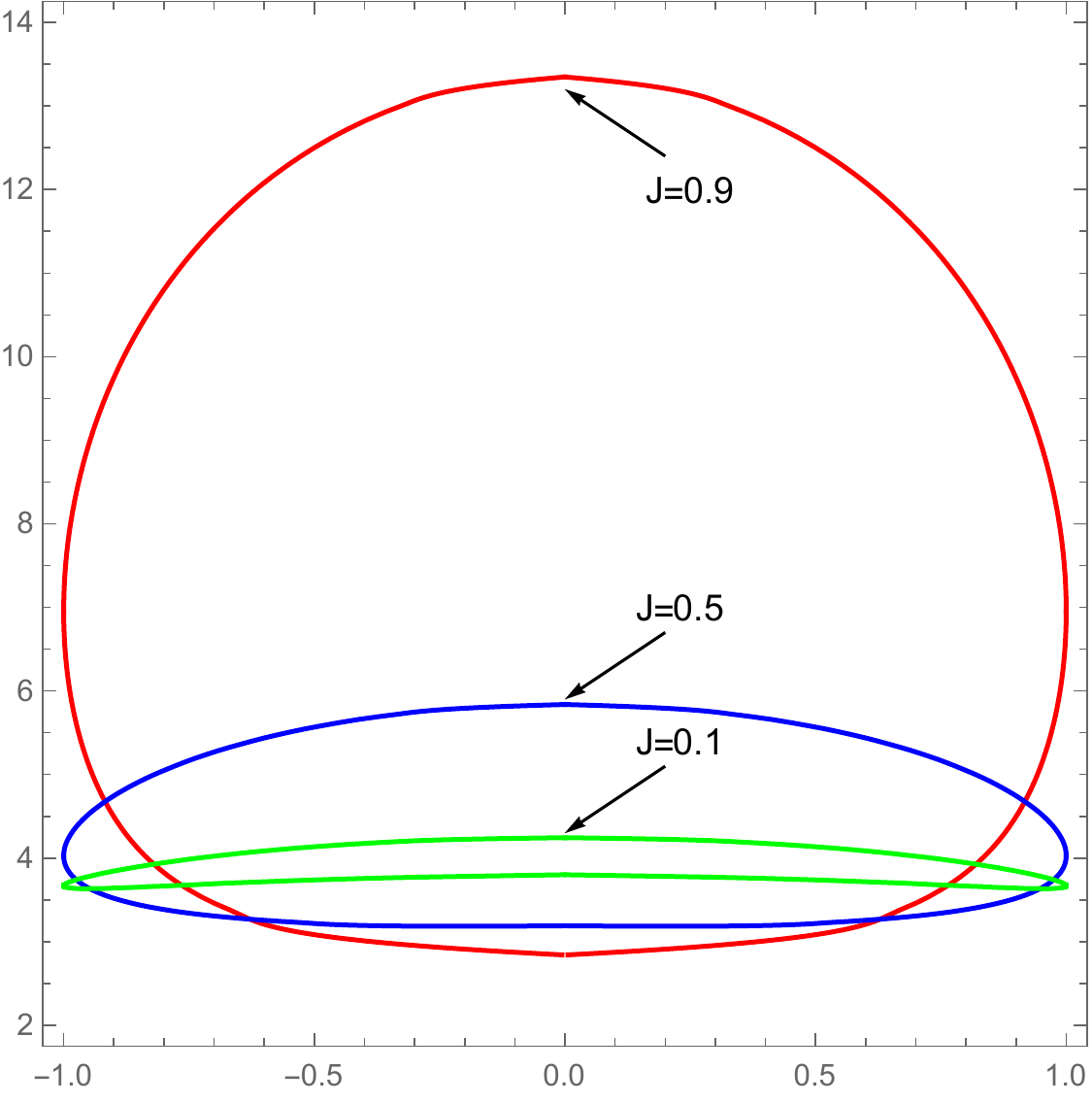} 
		\label{PRZK}
 }
\caption{$PR$ and $PF(\rho)$ in Tomimatsu-Sato metric ($M=1$, $J=0.1;0.5;0.9$).}.
\label{TS1}
\end{figure*}

\subsection{Tomimatsu-Sato}
\label{SS9}
The spacetime metric (\ref{a30}) is
\begin{align}
&\alpha^2=\frac{A}{B}+\frac{16q^2(1-y^2)C^2}{BD}, \quad  \gamma^2=\frac{\sigma^2(1-y^2)D}{p^2B},  \quad \omega=\frac{4 p q C}{D}, \\
&\psi^2=\frac{B}{p^4(x^2-y^2)^3(x^2-1)}, \quad  \phi^2=\frac{B}{p^4(x^2-y^2)^3(1-y^2)}.
\end{align}
Here the polynomial function $A,B,C,D$  are rather cumbersome and can be found, for example, in \cite{Kodama:2003ch}, where a detailed analysis of this metric is also carried out. The ADM mass and the angular momentum are $M=2\sigma /p$ and $J=M^2q$, respectively. The TS solution has an important feature - the area of causality violation in the external region $x>1$. As in Kerr and Zipoy-Voorhees, the causal region $P_\rho$ contains both singularity and spatial infinity if and only if $\rho_{min}<\rho<\rho_{max}$. However, now the maximum and minimum values do not have a simple analytical expression and are obtained from the formulas (\ref{a8}) numerically. 

$PR$ and $PF(\rho)$ are obtained by the same method as before, using the equation (\ref{a5}), and the boundary conditions (\ref{a1}) on the boundary of the causal region $\partial P_\rho$. Note that in the field of causality violation we must use the general equation (\ref{a10}) instead of (\ref{a5}), however, in the examples considered by us, such a need does not arise. The resulting solution for $PR$ and $PF(\rho)$ is shown in the Figs. \ref{TS1}. It is essential that the fundamental photon  orbits are non-spherical, thus the solution of the TS is also a non-integrable dynamical system \cite{Pappas:2018opz, Glampedakis:2018blj}.

$PR$ and $PF(\rho)$ for small rotations $J = 0.1$ resembles the $PR$ and $PF(\rho)$ for Zipoy-Voorhees solutions Fig. \ref{PRTS1}, \ref{PRTS3D}. For each allowed value of $\rho_{min}<\rho<\rho_{max}$ there is a well-defined fundamental photon hypersurface that is curved downward. Thus, like Zipoy-Voorhees, the solution will induce a set of relativistic images \cite{Virbhadra:2008ws}. The function $PF(\rho)$ in this case interpolates between the univalent and two-fold types and is purely “four-dimensional”. As a result, the minimum size of the shadow may not fall either on the equatorial plane or on the vertical axis of the shadow between them, since the fundamental photon surface is minimal for some intermediate value of the impact parameter. Therefore, the shadow will be a slightly asymmetric analog of the flattened shadow of Zipoy-Voorhees \cite{Abdikamalov:2019ztb, Galtsov:2019fzq}.

For intermediate values of rotation parameter $J=0.5$, the deformation of fundamental photon hypersurfaces decreases Figs. \ref{PRZ1}, \ref{PRZ2}, and the function $PF(\rho)$ is practically univalent. For large rotations,  $J = 0.9$, the fundamental photon hypersurfaces approaches the `` spherical '' ones $x={\rm const}$, and $PF(\rho)$ is univalent, which more closely resembles the Kerr solution Figs. \ref{PRZ13D}, \ref{PRZ23D}, in particular the minimum and maximum size of the shadow falls on the equatorial plane as it was in Kerr. Thus, we can expect that the shadow in this case will more closely resemble that in Kerr. These properties of the shadow of the solution of the TS really take place, as was demonstrated in \cite{Bambi:2010hf},  confirming  effectiveness of the  geometric constructions introduced here.
 
\setcounter{equation}{0}

\section{Conclusion}
\label{Concl}
This paper defines new geometric notions - the fundamental photon hypersurface and the fundamental photon region, generalizing the notion of the classical photon surface and the photon region to the case of stationary axially symmetric spaces with a complex, generically non-integrable, geodesic structure. They are based on the restriction of the umbilical condition on a certain naturally defined submanifold of the phase space  \cite{Cederbaum:2019vwv}.

These notions naturally complement the concept of fundamental photon orbits, supplying them with new geometric interpretation and the mathematical tools of the geometry of submanifolds. We formulate key theorems on the connection of the introduced hypersurfaces with the behavior of beams of null geodesics and derive structural equations for the principal curvatures of their spatial sections. We hope that the geometric objects and the formalism introduced by us will open the way for obtaining new topological restrictions, Penrose-type inequalities\cite{Shiromizu:2017ego,Feng:2019zzn,Yang:2019zcn}, uniqueness theorems\cite{Cederbaum,Yazadjiev:2015hda,Yazadjiev:2015mta,Yazadjiev:2015jza,Rogatko,Cederbaumo}, similar to ones for photon spheres and transversaly trapping surfaces \cite{Yoshino1,Yoshino:2019dty,Yoshino:2019mqw}. In particular, for the hypersurface $S_0$ under some additional assumptions we established the spherical topology $\mathbb S^2$.

In the second part of the paper, we introduced the concept of the fundamental photon function $PF(\rho)$ whose image is a classical photon region, and illustrated the application of our technique on the examples of Kerr, Zipoy -Voorhees and Tomimatsu-Sato in spheroidal coordinates. We found, in particular, that $PF(\rho)$ is some smooth function, which for small rotation  parameter resembles $PF(\rho) $ Zipoy-Voorhees ($PF(\rho)$ is two-sheeted), and for large rotation parameter -- Kerr solution ($PF(\rho)$ is univalent) and accordingly has an intermediate structure giving a new geometric justification of the optical properties of the shadow of the solution \cite{Bambi:2010hf}.

\begin{acknowledgements}
The work was supported by the Russian Foundation for Basic Research on the project 19-32-90095.
\end{acknowledgements}
  


\begin{thebibliography}{lab}

\bibitem{Virbhadra:1999nm} K.~S.~Virbhadra and G.~F.~R.~Ellis, ``Schwarzschild black hole lensing,'' 
Phys.\ Rev.\ D {\bf 62}, 084003 (2000).

\bibitem{Claudel:2000yi} C.~M.~Claudel, K.~S.~Virbhadra and G.~F.~R.~Ellis,
``The Geometry of photon surfaces,'' 
J.\ Math.\ Phys.\ {\bf 42}, 818 (2001).

\bibitem{Gibbons}
  G.~W.~Gibbons and C.~M.~Warnick,
  ``Aspherical Photon and Anti-Photon Surfaces,''
  Phys.\ Lett.\ B {\bf 763}, 169 (2016)
  [arXiv:1609.01673 [gr-qc]]

\bibitem{Chen} 
   B.~Y.~Chen, ``Pseudo-Riemannian Geometry, $\delta$-Invariants and Applications''

\bibitem{Okumura}
M.~Okumura,
``Totally umbilical hypersurfaces of a locally product Riemannian manifold,''
 Kodai Math.\ Sem. Rep.\ {\bf19},  35 (1967).

\bibitem{Senovilla:2011np} 
  J.~M.~M.~Senovilla,
  ``Umbilical-Type Surfaces in Spacetime,''
  [arXiv:1111.6910 [math.DG]].

\bibitem{Wilkins:1972rs} 
  D.~C.~Wilkins,
  ``Bound Geodesics in the Kerr Metric,''
  Phys.\ Rev.\ D {\bf 5}, 814 (1972).

\bibitem{Teo} 
  E.~Teo,
  ``Spherical photon orbits around a Kerr black hole,''
  Gen.\ Rel.\ Grav {\bf 35}, 1909 (2003). 

\bibitem{Grover:2017mhm}
  J.~Grover and A.~Wittig,
  ``Black Hole Shadows and Invariant Phase Space Structures,''
  Phys.\ Rev.\ D {\bf 96} (2017) no.2,  024045
  [arXiv:1705.07061 [gr-qc]].

\bibitem{Grenzebach}
  A.~Grenzebach, V.~Perlick and C.~Lammerzahl,
 ``Photon Regions and Shadows of Kerr-Newman-NUT Black Holes with a Cosmological Constant,''
  Phys.\ Rev.\ D {\bf 89}, no. 12, 124004 (2014)
  [arXiv:1403.5234 [gr-qc]].
 
\bibitem{Grenzebach:2015oea} 
  A.~Grenzebach, V.~Perlick and C.~Lammerzahl,
  ``Photon Regions and Shadows of Accelerated Black Holes,''
  Int.\ J.\ Mod.\ Phys.\ D {\bf 24}, no. 09, 1542024 (2015)
  [arXiv:1503.03036 [gr-qc]].

\bibitem{Cunha:2018acu}
  P.~V.~P.~Cunha and C.~A.~R.~Herdeiro,
  ``Shadows and strong gravitational lensing: a brief review,''
  Gen.\ Rel.\ Grav.\  {\bf 50} (2018) no.4,  42
  [arXiv:1801.00860 [gr-qc]].

\bibitem{Shipley:2019kfq}
  J.~O.~Shipley,
  ``Strong-field gravitational lensing by black holes,''
  [arXiv:1909.04691 [gr-qc]].
  
  \bibitem{Cederbaum}
  C.~Cederbaum and G.~J.~Galloway,
  ``Uniqueness of photon spheres in electro-vacuum spacetimes,''
  Class.\ Quant.\ Grav.\  {\bf 33}, 075006 (2016)
  [arXiv:1508.00355 [math.DG]].


\bibitem{Yazadjiev:2015hda} 
  S.~S.~Yazadjiev,
  ``Uniqueness of the static spacetimes with a photon sphere in Einstein-scalar field theory,''
  Phys.\ Rev.\ D {\bf 91}, no. 12, 123013 (2015)
  [arXiv:1501.06837 [gr-qc]].
  
\bibitem{Yazadjiev:2015jza}
  S.~Yazadjiev and B.~Lazov,
  ``Uniqueness of the static Einstein–Maxwell spacetimes with a photon sphere,''
  Class.\ Quant.\ Grav.\  {\bf 32}, 165021 (2015)
  [arXiv:1503.06828 [gr-qc]].
  
\bibitem{Yazadjiev:2015mta} 
  S.~Yazadjiev and B.~Lazov,
  ``Classification of the static and asymptotically flat Einstein-Maxwell-dilaton spacetimes with a photon sphere,''
  Phys.\ Rev.\ D {\bf 93}, no. 8, 083002 (2016)
  [arXiv:1510.04022 [gr-qc]].
 
\bibitem{Rogatko}
  M.~Rogatko,
  ``Uniqueness of photon sphere for Einstein-Maxwell-dilaton black holes with arbitrary coupling constant,''
  Phys.\ Rev.\ D {\bf 93}, no. 6, 064003 (2016)
  [arXiv:1602.03270 [hep-th]].

\bibitem{Cederbaumo}
  C.~Cederbaum,
  ``Uniqueness of photon spheres in static vacuum asymptotically flat spacetimes,''
  [arXiv:1406.5475 [math.DG]].
  
\bibitem{Yoshino:2016kgi} 
  H.~Yoshino,
  ``Uniqueness of static photon surfaces: Perturbative approach,''
  Phys.\ Rev.\ D {\bf 95}, no. 4, 044047 (2017)
  [arXiv:1607.07133 [gr-qc]].
  
\bibitem{Pappas:2018opz} 
  G.~Pappas and K.~Glampedakis,
  ``On the connection of spacetime separability and spherical photon orbits,''
  [arXiv:1806.04091 [gr-qc]].
   
\bibitem{Glampedakis:2018blj} 
  K.~Glampedakis and G.~Pappas,
  ``The modification of photon trapping orbits as a diagnostic of non-Kerr spacetimes,''
  [arXiv:1806.09333 [gr-qc]]. 

\bibitem{Galtsov:2019fzq} 
  D.~V.~Gal'tsov and K.~V.~Kobialko,
  ``Photon trapping in static axially symmetric spacetime,''
  Phys.\ Rev.\ D {\bf 100}, no. 10, 104005 (2019)
  [arXiv:1906.12065 [gr-qc]].
  
\bibitem{Galtsov:2019bty} 
  D.~V.~Gal'tsov and K.~V.~Kobialko,
  ``Completing characterization of photon orbits in Kerr and Kerr-Newman metrics,''
  Phys.\ Rev.\ D {\bf 99}, no. 8, 084043 (2019)
  [arXiv:1901.02785 [gr-qc]].

\bibitem{Cunha:2017eoe} 
  P.~V.~P.~Cunha, C.~A.~R.~Herdeiro and E.~Radu,
  ``Fundamental photon orbits: black hole shadows and spacetime instabilities,''
  Phys.\ Rev.\ D {\bf 96}, no. 2, 024039 (2017)
  [arXiv:1705.05461 [gr-qc]].
  
\bibitem{Cornish:1996de} 
  N.~J.~Cornish and G.~W.~Gibbons,
  ``The Tale of two centers,''
  Class.\ Quant.\ Grav.\  {\bf 14}, 1865 (1997).

\bibitem{Cunha:2016bjh} 
  P.~V.~P.~Cunha, J.~Grover, C.~Herdeiro, E.~Radu, H.~Runarsson and A.~Wittig,
  ``Chaotic lensing around boson stars and Kerr black holes with scalar hair,''
  Phys.\ Rev.\ D {\bf 94}, no. 10, 104023 (2016)
  [arXiv:1609.01340 [gr-qc]].

\bibitem{Semerak:2012dw}
  O.~Semerak and P.~Sukova,
  ``Free motion around black holes with discs or rings: between integrability and chaos - I,''
  Mon.\ Not.\ Roy.\ Astron.\ Soc.\  {\bf 404} (2010) 545
  [arXiv:1211.4106 [gr-qc]].

\bibitem{Shipley:2016omi} 
  J.~Shipley and S.~R.~Dolan,
  ``Binary black hole shadows, chaotic scattering and the Cantor set,''
  Class.\ Quant.\ Grav.\  {\bf 33}, no. 17, 175001 (2016)
  [arXiv:1603.04469 [gr-qc]].

\bibitem{Cunha:2018gql}
  P.~V.~P.~Cunha, C.~A.~R.~Herdeiro and M.~J.~Rodriguez,
  ``Does the black hole shadow probe the event horizon geometry?,''
  Phys.\ Rev.\ D {\bf 97} (2018) no.8,  084020
  [arXiv:1802.02675 [gr-qc]].
  
\bibitem{Cederbaum:2019vwv}
  C.~Cederbaum and S.~Jahns,
  ``Geometry and topology of the Kerr photon region in the phase space,''
  Gen.\ Rel.\ Grav.\  {\bf 51} (2019) no.6,  79
  [arXiv:1904.00916 [math.DG]].
  
\bibitem{Zipoy} 
D.~M.~Zipoy, J. Math. Phys. {\bf 7}, 1137 (1966)

\bibitem{Voorhees:1971wh} 
  B.~H.~Voorhees,
  ``Static axially symmetric gravitational fields,''
  Phys.\ Rev.\ D {\bf 2}, 2119 (1970).

\bibitem{Griffiths}  J.~B.~Griffiths and J.~Podolsky, ``Exact Space-Times in Einstein's General Relativity.''\ Cambridge University Press, \ 2009.

\bibitem{Kodama:2003ch} H.~Kodama and W.~Hikida,
``Global structure of the Zipoy-Voorhees-Weyl spacetime and the delta=2 Tomimatsu-Sato spacetime,''
Class.\ Quant.\ Grav.\ {\bf 20}, 5121 (2003)

\bibitem{Bambi:2010hf} 
  C.~Bambi and N.~Yoshida,
  ``Shape and position of the shadow in the $\delta = 2$ Tomimatsu-Sato space-time,''
  Class.\ Quant.\ Grav.\  {\bf 27}, 205006 (2010)
  [arXiv:1004.3149 [gr-qc]].


\bibitem{Abdikamalov:2019ztb}
  A.~B.~Abdikamalov, A.~A.~Abdujabbarov, D.~Ayzenberg, D.~Malafarina, C.~Bambi and B.~Ahmedov,
  ``A black hole mimicker hiding in the shadow: Optical properties of the $\gamma$ metric,''
  arXiv:1904.06207 [gr-qc].

  \bibitem{Lukes} G. Lukes-Gerakopoulos, 
``The non-integrability of the Zipoy-Voorhees metric,'' 
Phys.\ Rev. \ D {\bf 86}. (2012)  [arXiv:1206.0660].

\bibitem{Cao:2019vlu} 
  L.~M.~Cao and Y.~Song,
  ``Quasi-local photon surfaces in general spherically symmetric spacetimes,''
  arXiv:1910.13758 [gr-qc].

\bibitem{Cederbaum:2019rbv} 
  C.~Cederbaum and G.~J.~Galloway,
  ``Photon surfaces with equipotential time-slices,''
  arXiv:1910.04220 [math.DG].
  
\bibitem{Yoshino1}
  H.~Yoshino, K.~Izumi, T.~Shiromizu and Y.~Tomikawa,
  ``Extension of photon surfaces and their area: Static and stationary spacetimes,''
  PTEP {\bf 2017}, no. 6, 063E01 (2017)
  [arXiv:1704.04637 [gr-qc]].

\bibitem{Yoshino:2019dty}
  H.~Yoshino, K.~Izumi, T.~Shiromizu and Y.~Tomikawa,
  ``Transversely trapping surfaces: Dynamical version,''
  arXiv:1909.08420 [gr-qc].

\bibitem{Yoshino:2019mqw}
  H.~Yoshino, K.~Izumi, T.~Shiromizu and Y.~Tomikawa,
  ``Formation of dynamically transversely trapping surfaces and the stretched hoop conjecture,''
  arXiv:1911.09893 [gr-qc].

\bibitem{Shiromizu:2017ego}
  T.~Shiromizu, Y.~Tomikawa, K.~Izumi and H.~Yoshino,
  ``Area bound for a surface in a strong gravity region,''
  PTEP {\bf 2017} (2017) no.3,  033E01
  [arXiv:1701.00564 [gr-qc]].

\bibitem{Feng:2019zzn}
  X.~H.~Feng and H.~Lu,
  ``On the Size of Rotating Black Holes,''
  [arXiv:1911.12368 [gr-qc]].
  
\bibitem{Yang:2019zcn}
  R.~Q.~Yang and H.~Lu,
  ``Universal bounds on the size of a black hole,''
  [arXiv:2001.00027 [gr-qc]].

\bibitem{Koga:2019uqd}
  Y.~Koga and T.~Harada,
  ``Stability of null orbits on photon spheres and photon surfaces,''
  Phys.\ Rev.\ D {\bf 100} (2019) no.6,  064040
  [arXiv:1907.07336 [gr-qc]].


\bibitem{Paganini:2016pct} 
  C.~F.~Paganini, B.~Ruba and M.~A.~Oancea,
  ``Characterization of Null Geodesics on Kerr Spacetimes,''
  [arXiv:1611.06927 [gr-qc]].

\bibitem{Kubiznak:2007kh} 
  D.~Kubiznak and P.~Krtous,
  ``On conformal Killing-Yano tensors for Plebanski-Demianski family of solutions,''
  Phys.\ Rev.\ D {\bf 76}, 084036 (2007)
  [arXiv:0707.0409 [gr-qc]].


\bibitem{Virbhadra:2008ws} 
  K.~S.~Virbhadra,
  ``Relativistic images of Schwarzschild black hole lensing,''
  Phys.\ Rev.\ D {\bf 79}, 083004 (2009)

 \bibitem{Virbhadra:2002ju}
  K.~S.~Virbhadra and G.~F.~R.~Ellis,
  ``Gravitational lensing by naked singularities,''
  Phys.\ Rev.\ D {\bf 65}, 103004 (2002).

\end{thebibliography}
\end {document}